\def\IsLLNCS{} 

\ifdefined\IsLLNCS
\documentclass{llncs}
\else
\documentclass[11pt]{article}
\usepackage{fullpage}
\fi

\usepackage{amsfonts,amsmath,amssymb,boxedminipage,color,url,nccmath}
\usepackage[normalem]{ulem}

\usepackage[mathscr]{euscript}
\usepackage{scalerel}

\usepackage[bottom]{footmisc} 

\usepackage{booktabs}
\usepackage{pifont}

\usepackage[pageanchor=false,pdfstartview=FitH,colorlinks,linkcolor=blue,filecolor=blue,citecolor=blue,urlcolor=blue]{hyperref} 

\ifdefined\IsLLNCS

\else\fi
\usepackage{amsthm} 

\usepackage[T1]{fontenc}
\usepackage{lmodern}
\AtBeginDocument{%
  \DeclareFontShape{T1}{lmr}{m}{scit}{<->ssub*lmr/m/scsl}{}%
}

\usepackage{enumitem} 
\usepackage[hypcap=false]{caption}

\usepackage{aliascnt}
\usepackage[numbers]{natbib} 
\usepackage{cleveref}
\usepackage{xspace}
\usepackage{xstring}

\usepackage{subcaption}
\usepackage{graphicx}

\usepackage{tikz}
\usetikzlibrary{positioning}

\ifdefined\IsLLNCS
\else
\usepackage[nottoc]{tocbibind} 
\fi

\usepackage{mathtools}

\newcommand{\remove}[1]{}

\newcommand{\TLLNCS}[2]{\ifdefined\IsLLNCS#1\else #2 \fi}

\ifdefined\IsDraft
    \newcommand{\authnote}[2]{\textbf{[{\color{red} #1's Note:} {\color{blue} #2}]}}
		
		\newcommand{\deleted}[1]{{\color{blue} ~Deleted:~{\color{red} #1}}}
		\newcommand{\TODO}[1]{{\color{red} TODO:} {\color{blue} #1}}
\else
    \newcommand{\authnote}[2]{}
		
		\newcommand{\deleted}[1]{}
		\newcommand{\TODO}[1]{}
		\renewcommand{\sout}[1]{}
\fi

\newcommand{\sdotfill}{\textcolor[rgb]{0.8,.8,0.8}{\dotfill}} 

\newenvironment{algorithm}{\vspace{-0.8\topsep}\medskip\noindent\vspace{-0.8\topsep}\sdotfill\begin{algo}}{\vspace{-\topsep}\sdotfill\end{algo}}

\newcommand \mycaption {\small }     
\newcommand \mylabel {}

\ifdefined\IsLLNCS
    \newenvironment{nfbox}[3]{
    \renewcommand \mycaption {#1}
    \renewcommand \mylabel {#2}
    \begin{center}\small
    \begin{tabular}{|ll|}
    \hline
    \hspace{.3ex}
    \begin{minipage}{.97\linewidth}
         \vspace{0.5ex}
         #3}
         {\smallskip
         \captionof{figure}{\mycaption}
         \label{\mylabel}
     \end{minipage}
     &\hspace{.3ex} \\
     \hline
     \end{tabular}
     \end{center}    
    }
\else
    
\fi

\ifdefined\IsLLNCS

\else

\fi

\newcommand{\ie}  {i.e.,\xspace}
\newcommand{\eg}  {e.g.,\xspace}

\newcommand{\wrt} {with respect to\xspace}
\newcommand{\wlg} {without loss of generality\xspace}

\newcommand{\abs}[1]{\left\lvert #1 \right\rvert}
\newcommand{\sabs}[1]{\lvert #1 \rvert}

\newcommand{\iprod}[1]{\langle #1 \rangle}
\newcommand{\set}[1]{\left\{#1\right\}}
\newcommand{\sset}[1]{\{#1\}}
\newcommand{\paren}[1]{\left(#1\right)}
\newcommand{\sparen}[1]{(#1)}

\newcommand{\tth}{\ensuremath{t}\textsuperscript{th}\xspace}

\newcommand{\NN}{{\mathbb{N}}}

\newcommand{\zo}{\{0,1\}}
\newcommand{\zon}{{\zo^n}}
\newcommand{\zom}{{\zo^m}}

\newcommand{\zos}{{\zo^\ast}}
\newcommand{\from}{\leftarrow}

\newcommand{\poly}{\operatorname{poly}}

\renewcommand{\cref}{\Cref}

\ifdefined\IsLLNCS
\newaliascnt{claiml}{theorem}
\newtheorem{claiml}[claiml]{Claim}
\aliascntresetthe{claiml}
\renewenvironment{claim}{\begin{claiml}}{\end{claiml}}
\crefname{claiml}{Claim}{Claims}
\else

\newtheorem{theorem}{Theorem}[section]

\newaliascnt{lemma}{theorem}
\newtheorem{lemma}[lemma]{Lemma}
\aliascntresetthe{lemma}

\newaliascnt{claim}{theorem}
\newtheorem{claim}[claim]{Claim}
\aliascntresetthe{claim}

\newaliascnt{corollary}{theorem}

\aliascntresetthe{corollary}

\newaliascnt{proposition}{theorem}

\aliascntresetthe{proposition}

\newaliascnt{conjecture}{theorem}

\aliascntresetthe{conjecture}

\newaliascnt{definition}{theorem}
\newtheorem{definition}[definition]{Definition}
\aliascntresetthe{definition}

\newaliascnt{remark}{theorem}
\newtheorem{remark}[remark]{Remark}
\aliascntresetthe{remark}

\newaliascnt{example}{theorem}

\aliascntresetthe{example}
\fi

\crefname{theorem}{Theorem}{Theorems}
\crefname{lemma}{Lemma}{Lemmas}
\crefname{figure}{Figure}{Figures}
\crefname{claim}{Claim}{Claims}
\crefname{corollary}{Corollary}{Corollaries}
\crefname{proposition}{Proposition}{Propositions}
\crefname{conjecture}{Conjecture}{Conjectures}
\crefname{definition}{Definition}{Definitions}
\crefname{remark}{Remark}{Remarks}
\crefname{exmaple}{Example}{Examples}

\newaliascnt{proto}{theorem}

\newtheorem{proto}[proto]{Protocol}

\aliascntresetthe{proto}
\crefname{proto}{protocol}{protocols}

\newaliascnt{algo}{theorem}
\newtheorem{algo}[algo]{Algorithm}
\aliascntresetthe{algo}
\crefname{algo}{algorithm}{algorithms}

\newaliascnt{expr}{theorem}
\newtheorem{expr}[expr]{Experiment}
\aliascntresetthe{expr}
\crefname{experiment}{experiment}{experiments}

\newaliascnt{fact}{theorem}

\aliascntresetthe{fact}
\crefname{fact}{fact}{facts}

\newaliascnt{observation}{theorem}

\aliascntresetthe{observation}
\crefname{observation}{observation}{observations}

\def\FullBox{$\Box$}
\def\qed{\ifmmode\qquad\FullBox\else{\unskip\nobreak\hfil
\penalty50\hskip1em\null\nobreak\hfil\FullBox
\parfillskip=0pt\finalhyphendemerits=0\endgraf}\fi}

\TLLNCS{
\newenvironment{proofof}[1]{\begin{proof}[Proof of~#1]}{\end{proof}}
}{
}

\newcommand{\Ex}{{\mathrm E}}
\newcommand{\eex}[2]{\Ex_{#1}\left[#2\right]}
\newcommand{\ex}[1]{\Ex\left[#1\right]}
\newcommand{\seex}[2]{\Ex_{#1}[#2]}

\renewcommand{\Pr}{{\mathrm {Pr}}}
\newcommand{\pr}[1]{\Pr\left[#1\right]}
\newcommand{\ppr}[2]{\Pr_{#1}\left[#2\right]}

\newcommand{\ppt}{{\sc ppt}\xspace}

\newcommand{\cS}{\mathcal{S}}

\newcommand{\cI}{\mathcal{I}}

\newcommand{\cG}{\mathcal{G}}
\newcommand{\cF}{\mathcal{F}}

\newcommand{\conc}{||}

\newcommand{\Tableofcontents}{
\ifdefined\IsLLNCS \else
    \thispagestyle{empty}
    \pagenumbering{gobble}
    \clearpage
    \setcounter{tocdepth}{2}
    \tableofcontents
    \thispagestyle{empty}
    \clearpage
    \pagenumbering{arabic}
\fi
}

\newcommand{\vect}[1]{\myvec{#1}}

\newcommand{\vs}{\vect{s}}

\newcommand{\vy}{\vect{y}}

\newcommand{\adv}{\mathcal{A}}

\newcommand{\secParam}{n}

\mathchardef\mhyphen="2D

\newcommand{\su}{\subseteq}

\newcommand{\myvec}[1]{{\bf #1}}

\newcommand{\SD}{\operatorname{SD}}
\newcommand{\of}[1]{\paren{#1}}

\newcommand{\eps}{\varepsilon}

\newcommand{\Break}{\adv}
\newcommand{\keylen}{\lambda}
\renewcommand{\stretch}{r}
\newcommand{\inlen}{\mathsf{in}}
\newcommand{\outlen}{\mathsf{out}}
\newcommand{\numquery}{c}
\newcommand{\numinputquery}{m}
\newcommand{\numcallsdist}{d}
\newcommand{\oracle}{\mathcal{O}}

\newcommand{\iscore}{\mathsf{is}}
\newcommand{\rscore}{\mathsf{rs}}

\newcommand{\vz}{\myvec{z}}

\newcommand{\keyDom}{\zo^{\keylen}}

\newcommand{\iprob}{\mathsf{ip}}
\newcommand{\rprob}{\mathsf{rp}}

\newcommand{\funcSpace}{\cF}

\newcommand{\prf}{F}

\newcommand{\real}{\mathsf{real}}
\newcommand{\ideal}{\mathsf{ideal}}
\newcommand{\tilR}{\widetilde{R}}

\newcommand{\lastset}{\cI}
\newcommand{\recSteps}{\ell}

\newcommand{\chall}{y^*}
\newcommand{\st}{\mathsf{st}}

\newcommand{\Adv}{\operatorname{Adv}}

\newcommand{\true}{\mathsf{true}}

\newcommand{\bnote}[1]{\authnote{Bar}{#1}}

\newcommand{\Bnote}[1]{\bnote}

\title{Lower Bounds on Black-Box Constructions \ifdefined\IsLLNCS\else\\\fi of Pseudorandom Functions}
\ifdefined\IsAnon
    \author{}
    \date{}
    \ifdefined\IsLLNCS\institute{}\fi
\else
	\ifdefined\IsLLNCS
    \author{
        Bar Alon\inst{1}\thanks{Work was conducted while at Georgetown University.}
        \and Itai Dinur\inst{2,3}
        \and Muthuramakrishnan Venkitasubramaniam\inst{3}\\
        \email{alonbar08@gmail.com, dinuri@bgu.ac.il, mv783@georgetown.edu}
    }
    \institute{School of Electrical Engineering, Tel-Aviv University. \and Department of Computer Science, Georgetown University \and Department of Computer Science, Ben-Gurion University}
	\else
	\author{
        Bar Alon\thanks{School of Electrical Engineering, Tel-Aviv University. Work was conducted while at Georgetown University.}\\\texttt{alonbar08@gmail.com}
        \and Itai Dinur\thanks{Department of Computer Science, Ben-Gurion University and Department of Computer Science, Georgetown University}\\\texttt{dinuri@bgu.ac.il}
        \and Muthuramakrishnan Venkitasubramaniam\thanks{Department of Computer Science, Georgetown University}\\texttt{mv783@georgetown.edu}
    }
	\fi
\fi

\begin{document}

\maketitle

\begin{abstract}
In their seminal work, \citeauthor*{GGM86} [CRYPTO 1984] constructed a pseudorandom function (PRF) using a black-box access to a pseudorandom generator (PRG). When combined with Levin's domain extension technique, the GGM construction invokes the PRG $\omega(\log n)$ times, where $n$ denotes the input length to the PRG. To this day, no black-box construction achieving fewer calls is known.

Recently, \citeauthor*{BMM24} [CRYPTO 2024] showed that for a certain family of constructions, which they termed \emph{tree constructions}, the GGM construction is optimal. However, the basic challenge of whether a PRF can be built with just \emph{one invocation} of the PRG still remains open.

In this work, we consider fully black-box constructions of PRFs from PRGs, where both the construction and the reduction are required to be black-box, and the number of interactions the reduction makes with the adversary is independent of the number of oracle calls the adversary makes to its underlying function within each interaction. 
Our main result shows that no such construction can have $o(n/\log n)$ and $o(\inlen/\log\inlen)$ \emph{non-adaptive} calls to the PRG, where $\inlen$ is the input length of the PRF. 
This impossibility holds even for weak PRFs with one-bit output, where the adversary is restricted to making i.i.d. uniformly random queries. In addition, we prove a lower bound for weak PRFs with sufficiently long outputs that holds even when the construction is allowed to make adaptive queries to the PRG.
\end{abstract}

\Tableofcontents

\section{Introduction}
A pseudorandom function (PRF) family is a collection of functions that appear random to any efficient algorithm. PRFs are fundamental to cryptography, both in theory and in practice, and have found many applications in areas outside of cryptography, such as in complexity and learning theory. 

The seminal work of \citet*{GGM86} (henceforth the GGM construction) shows that a PRF can be constructed from \emph{any} pseudorandom generator -- an efficient algorithm that on a random input, outputs a longer string that appears random to any efficient algorithm. Specifically, they showed that given black-box access to a pseudorandom generator (PRG) $G:\zon\to\zo^{2n}$, one can construct a PRF $\prf=\set{f_k:\zo^{\inlen(n)}\to\zon}_{k\in\zon}$ for any polynomial $\inlen(n)$. The number of calls to $G$ in the GGM construction is $|x|$ for any input $x$. A variant of the GGM construction that relies on Levin’s domain extension technique \cite{Levin85} improves the number of calls to $\omega(\log n)$.\footnote{The final construction adds to the key a hash  $h:\zo^{\inlen(n)}\to\zo^{\omega(\log n)}$ from pair-wise independent hash function family, and applies the GGM construction to $h(x)$.} However, a fundamental question has remained open since it was first discovered---is the construction optimal in the number of PRG invocations? More precisely, the following remains open:
\begin{quote}
	\centering\emph{What is the minimal number of invocations of a (length-doubling) PRG required to build a PRF?}
\end{quote}

Recently, \citet{BMM24} made progress in answering this question. For a limited class of constructions that they refer to as \emph{tree constructions} and generalizes the GGM construction, they prove that the GGM construction is optimal. In fact, their result holds even when considering PRFs with one-bit output. However, the state of affairs for more general (i.e. not necessarily tree) constructions, still remains open. In fact, previous results do not rule out the possibility of building a PRF with a \emph{single call} to the PRG.

\subsection{Our Results}
{We take a step towards answering the above question. We show two results. Our first and main result considers the already challenging setting of designing PRFs via non-adaptive calls to the underlying PRG (\ie the construction generates all of its queries to the PRG before sending them). We demonstrate a lower bound on the number of calls to the underlying PRG oracle required for a fully black-box construction of a PRF. The impossibility is for a class of black-box reductions, which we call \emph{query-bounded} reductions, which we define below. For our second result, we provide a lower bound for PRFs with a long output, which additionally includes adaptive constructions.} 

\paragraph{Lower bounds for black-box constructions with non-adaptive calls.}
Our main result demonstrates the impossibility of building a PRF with input length $\inlen$ from a PRG with an arbitrary polynomial stretch, using $c$ \emph{non-adaptive} calls, where $c=o(n/\log n)$ and $c=o(\inlen/\log\inlen)$ via a fully black-box construction. We demonstrate this impossibility for a class of black-box reductions, which we call \emph{query-bounded} reductions. 

Roughly, a reduction is query-bounded if the number of interactions it makes with the adversary is upper-bounded by a polynomial that is independent of the adversary's query complexity (for all sufficiently successful adversaries). Recall that in the PRF setting, the adversary has access to a function oracle that it may call a bounded number of times. We focus on reductions whose number of interactions with the adversary may depend on the adversary’s success probability, but is independent of the number of oracle calls the adversary makes to the underlying function within each interaction. We believe this is a natural class of reductions and captures all reductions in the literature (see \cref{sec:subtleties}). In particular, it includes reductions whose number of interactions with the adversary decreases as the adversary’s success probability increases. It is an intriguing question whether one could extend the lower bound to all black-box reductions.

\begin{theorem}[Informal, impossibility for fully black-box construction]\label{thm:intro_fbb}
	There is no fully black-box construction of a PRF with input length $\inlen$ from a PRG using {query-bounded reductions} and $c$ non-adaptive calls, where $c=o(n/\log n)$ and $c=o(\inlen/\log\inlen)$.
\end{theorem}
{The formal statement and proof of \cref{thm:intro_fbb} appear in \cref{sec:fullyBB}. A few notes are in place. First, our proof works even for \emph{weak} PRFs, where the adversary can only make i.i.d.\ uniform queries to the oracle. Second, unlike our second result (see \cref{thm:intro_sbb} below), \cref{thm:intro_fbb} holds for PRFs with one-bit output. Finally, we argue that as a corollary of \cref{thm:intro_fbb}, there is no fully black-box construction of a PRF from a PRG using a query-bounded reduction and a \emph{constant number} of non-adaptive calls, regardless of the input length of the PRF. Indeed, for $\inlen=O(\log n)$ PRFs can be constructed unconditionally by simply letting the key be the truth table of the function (in particular, there is no reduction). On the other hand, for $\inlen=\omega(\log n)$, \cref{thm:intro_fbb} states that there is no fully black-box construction with a constant number of calls. Note that this resolves the question of the single call case for query-bounded reductions.}

\paragraph{A lower-bound for black-box constructions with long outputs.}
We show an additional lower bound for PRFs with long outputs. Unlike in \cref{thm:intro_fbb}, here the result holds for adaptive constructions and does not require assuming the reduction to be query-bounded.

\begin{theorem}[Informal, lower bounds for long output constructions]\label{thm:intro_sbb}
	There is no fully black-box construction of a PRF with output length $\outlen$ from a PRG with stretch $\stretch$ using at most $\frac{\outlen}{\max\of{\stretch,\omega(\log n)}}$ calls.
\end{theorem}

Stated differently, if $\stretch=O(\log n)$, then \cref{thm:intro_sbb} states that there is no black-box construction using $o(\outlen/\log n)$ calls to the PRG. On the other hand, if $\stretch=\omega(\log n)$, then there is no black-box construction using $o(\outlen/\stretch)$ calls. Additionally, similarly to \cref{thm:intro_fbb}, the proof of \cref{thm:intro_sbb} also applies to weak PRFs.

Let us compare \cref{thm:intro_sbb} to the GGM construction (note that in the GGM construction, the output of the PRF is of length $n$, making it natural to compare to). The GGM construction (combined with Levin's trick) is an $n$-bit output PRF from an $n$-bit stretch PRG using $\omega(\log n)$ calls to the PRG. The PRG can be constructed from a 1-bit stretch PRG using $O(n/\log n)$ calls using the Goldreich-Levin hardcore function theorem \cite{GL89}. Thus, the overall number of calls the PRF construction makes to a 1-bit stretch PRG is $\omega(n)$. By \cref{thm:intro_sbb}, the number of calls must be at least $n/\omega(\log n)$. This means that the GGM construction is optimal up to a factor of $\omega(\log n)$. The formal statement of \cref{thm:intro_sbb} appears in \cref{sec:LB_PRF}.

\subsection{Related Works}
Separation of cryptographic primitives is a fundamental question in cryptography that qualitatively measures the relative hardness of different cryptographic primitives. The seminal work of \citet{IR89} provided the first separations ruling out constructions of key exchange from one-way functions. The class of constructions ruled out in this work is referred to as a fully black-box construction (see \cite{RTV04} for a taxonomy on reductions) where the construction uses the underlying primitive as a black-box and the reduction is black-box. 

Within the framework of black-box separations, the \emph{meta-reduction} paradigm introduced in the work of \citet{BV98}, considers designing a (possibly inefficient) attacker $\adv$ that when invoked with the reduction breaks the challenge in the underlying security reduction, then emulating $\adv$ in a manner that is efficient and does not affect the interaction of the reduction and the challenger (see also \cite{BressonMV08,HaitnerRS09,FischlinS10,Pass11,Pass13}). We follow the meta-reduction paradigm in this work (see \cite{BBF13} and references therein). 

\citet{MV11,MV15} were the first to consider lower bounds on black-box constructions of PRFs from PRGs.\footnote{The paper considers the question of extending the stretch of a PRG; however, they noted in \cite{MV15} that their result applies to constructing PRFs from PRGs as well.} They ruled out single-call constructions where the output is one of the PRG's output bits. More recently, \citet{BMM24} both generalized the result of \cite{MV11,MV15}, and proved a lower bound on the number of calls in tree constructions, which naturally generalize the GGM construction. Additionally, they showed a lower bound on the key length required for black-box constructions.

A related question is that of constructing PRFs from non-adaptive PRFs (naPRF), where security is only guaranteed to hold against adversaries who make non-adaptive queries. \citet{BH12} showed a non-uniform construction of a PRF from a naPRF that uses a single call. They further showed a uniform construction using $\omega(1)$ calls. Related to our results, \citet{ST20} showed that certain fully black-box non-adaptive constructions of PRFs from naPRFs are impossible.

\subsection{Organization}
We give an overview of our techniques in \cref{sec:overview}. The preliminaries are given in \cref{sec:prelim}. Then, in \cref{sec:fullyBB} we prove our main result, showing a lower bound for constructions with non-adaptive calls. Finally, in \cref{sec:LB_PRF}, we prove our lower bound for PRFs with long outputs.

\section{Our Techniques}\label{sec:overview}
We now present our techniques. We start with showing the impossibility for non-adaptive black-box constructions.
\subsection{Lower Bounds for Fully Black-Box Constructions}

As a warm-up, we will rule out PRF constructions that make a single call to a PRG and then generalize this argument for multiple non-adaptive calls. 
\ifdefined\IsLLNCS
\subsubsection{Warm-Up: The Single Call Case}
\else
\subsubsection{Warm-Up: The Single Call Case}
\fi
Since we assume the construction makes one call to the PRG, we may write it as
$$f^{G}_k(x)=P\of{x,k,G(S(x,k))},$$
for efficiently computable functions $S$ and $P$. In the following, let $\keylen$ denote the length of the keys, and let $\stretch$ denote the stretch of $G$.

Standard black-box separation like the approaches in \cite{MV11,BMM24} construct an oracle $\oracle$ relative to which a PRG $G$ exists, and construct an oracle-aided algorithm $\Break^{\oracle}$ that breaks the security of $\prf$ using a polynomial number of calls to $\oracle$. 
We argue that it may be impossible to use such an approach to rule out black-box constructions of PRF from a PRG. 
Consider the following PRF construction from an $\stretch$-bit stretch PRG $G$:
$$f^G_{h,k}(x)=\iprod{G(h(x)),k}$$
where the key is a pair $(h,k)$, where $h$ is sampled from a pairwise independent hash function family, $k\from\zo^{n+\stretch}$, and $\langle\cdot,\cdot\rangle$ is the inner product computed modulo 2. If $G$ is instantiated via a random oracle, then, for any polynomial set of inputs $x_1,\ldots,x_m$, the set $\{h(x_1),\ldots,h(x_m)\}$ will all be distinct (except with negligible probability). This would imply that the entropy of each $G(x_i)$ is large enough (given the other outputs of $f$) and the inner product can extract a random bit, such that the $\numinputquery$ bits are statistically close to a uniform sample from $\zo^{\numinputquery}$. In fact, this argument works even if $G$ is only ``mildly'' random, \eg $G$ 
applies a random function only to the first $\omega(\log n)$ bits. Since the proof of security of the PRF is information-theoretic and does not depend on the reduction, it may be impossible to rule out such constructions using random oracles and an adversary $\Break^{\oracle}$ that calls $\oracle$ a polynomial number of times. Since the work of \citet{GT00}, most black-box separations from a random oracle have pursued this approach.

We circumvent this challenge by relying on the meta-reduction paradigm. We first define an ``ideal'' adversary that breaks the PRF using an \emph{exponential number} of queries to the PRG. We then construct an \emph{inefficient} ``real'' adversary that can simulate the ideal adversary when given the queries made by the reduction to the PRG. The idea here is that such an adversary can be simulated together with the reduction.  
In other words, we construct  an algorithm that breaks the security of the PRG by emulating the reduction with real adversary that together make only polynomial queries to the underlying PRG.

\paragraph{A naive attempt.} Let us start with a naive construction for the real and ideal adversaries that \emph{does not} work.
Define the ideal adversary $\adv_{\ideal}^{G,f}$ to query $G$ on all inputs, and query the PRF/random function $f$ on $\numinputquery$ distinct inputs $x_1,\ldots,x_\numinputquery$. The adversary outputs 1, indicating PRF, if and only if it can find a key that is consistent with all the responses. 
Clearly, if $f$ is sampled from the PRF family, $\adv_\ideal$  will always output 1. When $f$ is sampled from the random function family, we can show that it will output 1 with negligible probability, assuming $\numinputquery$ is sufficiently large via an entropy argument. Since the key length is $\keylen$, if $m = \keylen + n$, then except with probability $2^{-n}$, there will not be a consistent key.

Now, consider the following ``real adversary'' that aims to emulate the ideal adversary. This adversary queries its oracle (whose answers are given by reduction) on $\numinputquery$ inputs just like the ideal adversary, but along with the responses from the reduction, it also receives all queries made by the reduction to the PRG. The idea is that in order for the responses to be consistent with the PRF family, the reduction must have queried the PRG on all the points needed to emulate the PRF for some key $k$. Thus, using the query-answers provided to the real adversary, we let it (in exponential time) find the key $k$ without making any additional queries to the PRG itself. If it finds such a key, the adversary outputs 1 indicating PRF and outputs 0 otherwise. 

This adversary seemingly emulates the ideal adversary.  However, this is not the case. 
The reduction does not have to make all the queries to the PRG consistent with a particular key to compute the answers. In fact, in the GGM reduction, the adversary implants the challenge string as the output of a PRG call. This means that the above strategy for the real adversary fails, because the reduction may have computed some of the outputs using some key and querying $G$, but guessed the other outputs (without querying $G$). Therefore, the real adversary cannot find a consistent key, while the ideal adversary may find one.

Our first observation is that the reduction cannot guess too many of the outputs, while still being consistent with some key. Indeed, otherwise, the probability that it is consistent with some key becomes negligibly small. Thus, we can revise the strategy of the two adversaries to only look for a key that is consistent with, say, 0.8 fraction\footnote{The choice of 0.8 is arbitrary; any constant larger than 1/2 and smaller than 1 works.} of the inputs. 
Indeed, when using 0.8 as the threshold, if the ideal adversary interacts with a random function $f$, then for every key, on expectation, only 1/2 of the answers are consistent with the inputs. By a simple concentration bound, it will not find a key except with negligible probability.

Unfortunately, there is the issue of the reduction guessing too \emph{few} of the outputs. Specifically, in the case where the reduction computed correctly only 0.7 fraction of the inputs, we do not know how to show that the ideal adversary will not find a \emph{different} key that is consistent with 0.8 fraction of the inputs.

We introduce two ideas to solve this. Towards describing our ideas, we first introduce two \emph{score functions} (one for each adversary). For the ideal adversary, the score function \wrt some key is simply the number of consistent input-output pairs. 
The real adversary estimates the score function of the ideal adversary by considering the probability that it will be consistent with the answers (instead of a binary value). Specifically, for every key $k$ the real score is computed as follows: for every $x$ that is one of the $\numinputquery$ inputs (and $S(x,k)$ has not been queried by the reduction), compute the probability over a random $G'$ that is consistent with the reduction queries, that $f^{G'}_{k}(x)$ equals the value given to it by the reduction. This probability can be computed exactly by the real adversary. The final score of the real adversary for $k$ is the sum of these probabilities over all inputs. By linearity of expectation, it is equal to the expected score of the ideal adversary (for the same key).
The final score of each of the real and ideal adversaries is
its maximal score over all keys, and the answer (PRF/random function) is defined by comparing this score to some threshold.

However, choosing a fixed threshold could fail in case it falls between the real and ideal scores (which could happen even if the scores are very close). Therefore, the second idea is to choose a \emph{random threshold} between 0.8 and 1. This helps to smoothen the difference and as long as the real and ideal scores are close enough.

However, there are still two issues with the above argument.
\begin{enumerate}
	\item Although for every $k$, the score of the real adversary equals the expected score of the ideal adversary, the final outputs of the PRF on the $\numinputquery$ inputs might be dependent. Hence, we cannot claim that the scores are close with high probability. For a key $k$, such dependencies occur if for two different inputs $x_i$ and $x_j$, $S(x_i,k)=S(x_j,k)$.

	\item The second issue relates to the challenge $\chall$ held by the reduction. Recall that $\chall$ is either uniform over $\zo^{n+\stretch}$, or equals to $G(s)$, where $s\from\zon$. In the latter case, if the reduction implant $\chall$ instead of $G(S(x_i,k))$ for some $i$, then this increases the ideal score by 1, but increases the real score only by a negligible amount. Otherwise, the challenge affects the scores the same, which may provide the reduction a way to distinguish between the two adversaries. 
\end{enumerate}

\paragraph{Fixing the issues.}
To solve the problems, we aim to break the dependencies between the outputs of the PRF. This solves the first issue as it allows us to use Chernoff's inequality to argue that the two scores are close with high probability. It also solves the second one since if the outputs are independent, the challenge can only affect at most one output (keeping the scores close).
Towards breaking the dependencies, we observe that for every key $k$, either there is a PRG query that appears for many inputs, or we can find a large set of inputs with distinct PRG seeds.

\begin{claim}[Informal]\label{clm:intro_large_ind_set}
	For a positive integer parameter $p$ (to be chosen), one of the following must hold for every key $k$.
	\begin{enumerate}
		\item\label{item:intro_many_reps} There exists a seed that appears at least $p$ times in $\set{S(x_i,k):i\in[\numinputquery]}$, or
		
		\item There is a set $\cI\su[\numinputquery]$ of at least $m/p$ inputs with distinct PRG seeds.
	\end{enumerate}
\end{claim} 
\begin{proof}
	Consider the graph $H=([\numinputquery],E)$, where $i,j$ is an edge if and only if $S(x_i,k)=S(x_j,k)$. If \cref{item:intro_many_reps} does not hold, then the maximum degree in $H$ is at most $p-1$. By greedily selecting the vertices, we conclude there is an independent set of size at least $m/p$. By definition, the PRG seeds associated with an independent set are distinct.
\end{proof}

Now, if for a key $k$ the second item holds, then our previous argument works when restricting to $\cI$: Since we instantiate the PRG with a random oracle, this implies that the corresponding outputs are independent random variables (over the PRG random oracle). This independence allows us to use a concentration bound to argue that if $m/p$ is sufficiently large, then the real adversary can predict the output of the ideal adversary with good accuracy. 
As for the other keys, we argue that there is only a negligible fraction of them. Indeed, if for too many keys there is a seed $s$ that appears at least $p$ times then, intuitively, for a random key the entropy in the output for those $p$ inputs is bounded by $\keylen+n+\stretch$. Thus, for a sufficiently large $p$, the PRF can be distinguished from random. As a result, we may allow both the real and ideal adversaries to \emph{ignore} such keys. 

With this, for any polynomial $u$, we can choose $\numinputquery$ so that for every call made by the reduction to the adversary, the statistical gap between the real and ideal adversary's outputs is $1/u$-close. However, recall that the reduction can call the adversary polynomially many times, say $v$. If we want the argue that the statistical gap for all the combined outputs are close we need $v/u$ to be small. In other words $u$ should be chosen after $v$. Since $\numinputquery$ depends on $u$, for our impossibility to hold, we need to assume that the number of calls made by the reduction $v$ is independent of $\numinputquery$. We believe this is a reasonable restriction and discuss this further in Section~\ref{sec:subtleties}.

\ifdefined\IsLLNCS
\subsubsection{The General Case.}
\else
\subsubsection{The General Case}
\fi
We now describe our ideas to generalize the single call case to $c$  non-adaptive calls, where $c=o(n/\log n)$ and $c=o(\inlen/\log\inlen)$. Recall that to define the two adversaries, we define two score functions (one for each adversary) that aim to compute for how many inputs the reduction returned the correct output (with respect to some key).
The two score functions are defined \wrt each key, via a procedure that depends only on the construction of the PRF (in particular, it is independent of the output of the PRG). This is important, since the score function does not require the real adversary to make any PRG queries.

As in the single call case, the goal of the functions is to find a large set of inputs that make disjoint PRG queries. 
At a high level, this is done by letting the score functions try to break the dependence between some PRF outputs by fixing PRG outputs on queries that are made by many PRF inputs. 
Specifically, the procedure works in iterations as we next describe.

In each iteration, the adversaries looks for a seed $s$ that appears for sufficiently many inputs (\ie for many of the $x_i$'s, at least one of the seeds needed to compute $f^G_k(x_i)$ is $s$). For such seeds, the PRG will not contribute much entropy to the output of the PRF (taken over all inputs that query $s$). When we define the scores, we replace $G(s)$ with a fixed value that maximizes the score and is \emph{independent} of $G$. In such a case, the construction is left with one less call to the PRG. This process continues until either no $s$ appear many times, or until the process finds $c$ such seeds. In the former case, we can apply (a generalization of) \cref{clm:intro_large_ind_set} to argue that there is a large set of inputs with a mutually disjoint set of seeds, while in the latter case, we will later replace all calls to the PRG with fixed values.

In more detail, for every key $k$, the final output of the procedure is a list of $\recSteps(k)\leq c$ seeds $s_1,\ldots,s_{\recSteps(k)}$ and a list of nested sets of inputs $\cI_1\supseteq\ldots\supseteq\cI_{\recSteps(k)}$. We require that each $s_t$ appears for sufficiently many inputs in $\cI_t$. The score calculation is then done by replacing each $G(s_t)$ with a fixed value. If $\recSteps(k)=c$, then the scores are fixed and independent of $G$. Otherwise, there is a set $\cI\su\cI_{\recSteps(k)}$ of inputs with mutually disjoint seeds. The score calculation is then done the same as in the single call case.

For the argument to work, we need to ensure that at every step, the set of inputs $\cI_t$ that are being considered is sufficiently large. Specifically, if the procedure stops early, then \cref{clm:intro_large_ind_set} implies the existence of a set of $\frac{|\cI_t|}{(c-t)|\cI_{t+1}|}$ inputs with mutually disjoint queries. For the concentration bound to work, this ratio needs to be sufficiently large (Specifically, it needs to be larger than some fixed polynomial that depends on the key length, the PRG stretch, and the reduction). If the procedure replaces every $G(s_t)$, we need $|\cI_{\recSteps(k)}|$ to be sufficiently large as well to argue that the PRF can be distinguished from a random function. This implies that the number of inputs required is $\numinputquery=(\poly(n))^c\cdot c!$. {For the adversaries to be well-defined, this value must be upper bounded by $2^{\inlen}$, which holds for all sufficiently large $n$'s by our assumption that $c=o(\inlen/\log\inlen)$. Additionally}, the total number of queries that $R$ makes to the PRG is $\numcallsdist\cdot\numinputquery^\beta$, for some constant $\beta\in\NN$ and where $\numcallsdist$ is the number of calls $R$ makes to the distinguishing adversary. Recall that $c=o(n/\log n)$, hence $\numcallsdist\cdot\numinputquery^\beta\leq 2^{n/4}$ for all sufficiently large $n$'s. Since we instantiate the PRG with a random function, $2^{n/4}$ queries are insufficient to break its security (see \cref{lem:RO_is_PRG} for the formal statement), resulting in a contradiction.

\ifdefined\IsLLNCS
\subsubsection{Subtleties With Oracle Access to PRF Distinguishers.}\label{sec:subtleties}
\else
\subsubsection{Subtleties With Oracle Access to PRF Distinguishers}\label{sec:subtleties}
\fi
The standard definition of a fully black-box reduction from PRF to PRG is as follows. Given access to an oracle-adversary $\adv^{(\cdot)}$ that is able to distinguish the PRF from a random function, the reduction, a \ppt oracle-aided algorithm $R$ is able to distinguish the output of $G(s)$, where $s\from\zon$, from a string uniformly sampled from $\{0,1\}^{|G(s)|}$.

Since we fixed $R$, its run-time is some fixed polynomial, say $T(n)$. The adversary $\adv^f$ is allowed to query $f$ a polynomial $\numinputquery$ number of inputs. When activated by the reduction, $\adv$ expects $\numinputquery(n)$ answers. These should be provided by $R$; however, $\numinputquery$ is chosen by $\adv$, which is fixed after $R$, hence $\numinputquery$ cannot depend on $T$. In particular, if $\numinputquery(n)>T(n)$, then $R$ has no time to answer $\adv$'s queries.

For the model to make sense, we have to treat the $\numinputquery$ inputs chosen by $\adv$ as part of the reduction's input. However, now the reduction depends on $\adv$. For example, after receiving the $\numinputquery$ inputs from $\adv$, the reduction can decide to call it $\numinputquery$ more times. If we now replace $\adv$ with another adversary that sends $\numinputquery+1$ inputs, and apply $\adv$ on the first $\numinputquery$ answers (\ie it ignores the last input and output), then $R$ calls it more times than with $\adv$, thus the entire procedure changed, even though the adversary does the same. Such reductions are arguably not ``black-box''. 

{
Recall that in our proof, we had to assume that the reduction cannot depend on the number of inputs given by the distinguishing oracle. This means that our proof does not hold for the standard notion of fully black-box constructions. Instead, we consider \emph{query-bounded} reductions.
These are reductions whose number of calls to the adversary is bounded by some polynomial $\numcallsdist(n)$ for all adversaries whose distinguishing advantage is sufficiently large. Specifically, for our proof, we assume that there exists a polynomial $\numcallsdist(n)$ such that for any adversary $\adv$ whose distinguishing advantage is at least $\alpha(n)=1-e^{-n}$, the reduction calls $\adv$ at most $\numcallsdist(n)$ times. 
Note that any reduction whose number of calls to $\adv$ \emph{decreases} as a function of the distinguishing advantage of $\adv$ is also query-bounded (perhaps with a different $\alpha(n)$). We argue that the class of query-bounded reductions is both natural and well-motivated. Indeed, all reductions previously considered in the literature, to the best of our knowledge, satisfy this property and thus fall within the class of query-bounded reductions. This includes, for example, the GGM construction \cite{GGM86}, \citet{GL89} hardcore predicate construction, and \citet{Waters05} construction of signature schemes. We refer the reader to \cref{def:bounded} for a formal definition of query-bounded reductions.}

\subsection{Lower Bounds for PRFs With Long Outputs}
We now prove \cref{thm:intro_sbb}. To simplify this introduction, we will prove a simpler result that considers constructions with a single call to the PRG (see \cref{sec:LB_PRF} for the proof of the general result). Although it is not as general, the proof highlights the key ideas of our techniques. Specifically, we show that if $\outlen-\stretch=\omega(\log n)$, then there is no black-box construction of a PRF with output length $\outlen$ from a PRG with stretch $\stretch$ using a single call to the PRG.

Let $\inlen=\inlen(n)$ and assume that $\prf=\sset{f_k^G:\zo^{\inlen}\to\zo^{\outlen}}_{n\in\NN,k\in\zo^\keylen}$ is a PRF construction that calls $G:\zo^n\to\zo^{n+r}$ once.
The main idea of the proof is to consider the entropy of sufficiently many applications of the function. That is, given access to a function $f$ (that is either truly random or sampled from $\prf$), we consider the total entropy of $\numinputquery$ applications of $f$ on different inputs, where $\numinputquery=\numinputquery(n)$ is a sufficiently large polynomial. Observe that for a random function the entropy is maximized, \ie it is $\outlen\cdot \numinputquery$. However, we show how to construct a PRG $G$ such that the entropy of $\numinputquery$ applications of $f^G_k$ is noticeably far from $\outlen\cdot\numinputquery$. Intuitively, this is achieved by considering a PRG that is only applied to the first $\outlen-\stretch-1=\omega(\log n)$ bits of the input. This ensures that the entropy in the output comes from the key $k$ and the $\outlen-1$ (pseudo)random bits obtained from the stretch of $G$. Thus, the entropy will be at most $\keylen+\numinputquery\cdot(\outlen-1)$. Taking $\numinputquery=\keylen+n$ results in a gap of $n$ from $\outlen\cdot\numinputquery$, which we show is sufficient.

We note that the above argument works for all fully black-box reductions, and not just query-bounded ones, as was needed in \cref{thm:intro_fbb}.
\section{Preliminaries}\label{sec:prelim}
\subsection{Notations}
For $n,m\in\NN$ we let $[n]=\{1,2\ldots n\}$ and let $\binom{[n]}{m}=\set{\cS\su[n]:|\cS|=m}$. All logarithms are taken in base 2. For a set $\cS$, we write $s\from\cS$ to indicate that $s$ is selected uniformly at random from $\cS$. Given a random variable (or a distribution) $X$, we write $x\from X$ to indicate that $x$ is selected according to $X$. \ppt stands for probabilistic polynomial time. A function $\mu:\NN\to[0,1]$ is called negligible, if for every positive polynomial $p(\cdot)$ and all sufficiently large $n$, $\mu(n)<1/p(n)$. For $n,m\in\NN$, we let $\funcSpace_{n,m}=\set{f:\zon\to\zo^m}$ denote the set of all functions from $\zon$ to $\zom$, and let $\funcSpace_n=\funcSpace_{n,1}$.

\begin{theorem}[Hoeffding's inequality]\label{thm:Hoeffding}
	Let $X_1,\ldots,X_n\in\zo$ denote independent random variables such that $\pr{X_i=1}=p_i$ for all $i\in[n]$ for some $p_i\in(0,1)$. Let $p=\sum_{i=1}^n p_i$. Then for every $t>0$,
	$$\pr{\sum_{i=1}^n X_i\geq pn+t}\leq e^{-2t^2/n}$$
	and
	$$\pr{\sum_{i=1}^n X_i\leq pn-t}\leq e^{-2t^2/n}.$$
\end{theorem}

\subsection{Pseudorandom Generators}
\begin{definition}[Pseudorandom generator (PRG)]
	An efficiently computable function $G$ is an $\stretch(n)$-bit stretch PRG if for every $n\in\NN$ and $s\in\zos$, $|G(s)|=|s|+\stretch(|s|)$ and for every \ppt algorithm $\adv$, there exists a negligible function $\mu$ such that
	$$\abs{\ppr{s\from\zon}{\adv(1^n,G(s))=1}-\ppr{y\from\zo^{n+\stretch(n)}}{\adv(1^n,y)=1}}\leq\mu(n).$$
\end{definition}

The following lemma states that a random function is a PRG with overwhelming probability.
\begin{lemma}\label{lem:RO_is_PRG}
	The following holds for every $\stretch:\NN\to\NN$, and all sufficiently large $n\in\NN$. Let $G:\zon\to\zo^{n+\stretch(n)}$ be a function chosen uniformly at random and $C$ be a $2^{n/4}$-call oracle aided (possibly inefficient) algorithm. Then with probability at least $1-2^{-n/2}$,
	$$\abs{\ppr{s\from\zon}{C^{G}(G(s))=1}-\ppr{y\from\zo^{n+\stretch(n)}}{C^{G}(y)=1}}\leq2^{-n/4}.$$
\end{lemma}
\begin{proof}
    By Markov's inequality, it suffices to show that
    $$\eex{G\from\funcSpace_{n,n+\stretch(n)}}{\abs{\ppr{s\from\zon}{C^{G}(G(s))=1}-\ppr{y\from\zo^{n+\stretch(n)}}{C^{G}(y)=1}}}\leq2^{-3n/4}.$$
    This is done by lazy sampling \cite{BR06}. Assume \wlg that $C$ makes only distinct queries to the oracle. Let $P_0:\zon\to\zo^{n+\stretch(n)}$ be the procedure defined as follows:
    \begin{enumerate}
        \item Sample $s\from\zon$ and $y\from\zo^{n+\stretch(n)}$, independently.
        \item On query $x\in\zon$, sample $y^*\from\zo^{n+\stretch(n)}$.
        \item If $s=x$, then set $bad=\true$. Set $y=y^*$.
        \item Return $y$.
    \end{enumerate}
    Additionally, let $P_1:\zon\to\zo^{n+\stretch(n)}$ be the procedure defined as follows:
    \begin{enumerate}
        \item Sample $s\from\zon$ and $y\from\zo^{n+\stretch(n)}$, independently.
        \item On query $x\in\zon$, sample $y^*\from\zo^{n+\stretch(n)}$.
        \item If $s=x$, then set $bad=\true$.
        \item Return $y$.
    \end{enumerate}
    First, note that
    \begin{align*}
        \eex{G\from\funcSpace_{n,n+\stretch(n)}}{\ppr{s\from\zon}{C^{G}(G(s))=1}}=\ppr{y^*\from\zo^{n+\stretch(n)}}{C^{P_0}(y^*)=1}
    \end{align*}
    and that
    \begin{align*}
        \eex{G\from\funcSpace_{n,n+\stretch(n)}}{\ppr{y\from\zo^{n+\stretch(n)}}{C^{G}(y)=1}}=\ppr{y\from\zo^{n+\stretch(n)}}{C^{P_1}(y)=1}.
    \end{align*}
    Therefore,
    \begin{align*}
        &\eex{G\from\funcSpace_{n,n+\stretch(n)}}{\abs{\ppr{s\from\zon}{C^{G}(G(s))=1}-\ppr{y\from\zo^{n+\stretch(n)}}{C^{G}(y)=1}}}\\
        &\quad\leq\abs{\ppr{y^*\from\zo^{n+\stretch(n)}}{C^{P_0}(y^*)=1}-\ppr{y\from\zo^{n+\stretch(n)}}{C^{P_1}(y)=1}}
    \end{align*}
    
    Now, observe that the procedures $P_0$ and $P_1$ are identical-until-$bad$, namely, they return identically distributed answers until $bad$ occurs. Since $C$ can make $2^{n/4}$ queries, $bad$ occurs with probability at most $2^{n/4-n}=2^{-3n/4}$. Therefore,
    \begin{align*}
        \abs{\ppr{y^*\from\zo^{n+\stretch(n)}}{C^{P_0}(y^*)=1}-\ppr{y\from\zo^{n+\stretch(n)}}{C^{P_1}(y)=1}}
        \leq2^{-3n/4}.
    \end{align*}
\end{proof}

\subsection{Black-Box Constructions of Pseudorandom Functions}
We next define black-box constructions of a pseudorandom function (PRF) from PRGs.

\begin{definition}[Fully black-box PRF constructions]\label{def:fbb}
	A \emph{fully black-box construction} of a PRF from an $\stretch(n)$-bit stretch PRG is a pair $(\prf,R)$ where $$\prf=\set{f^{(\cdot)}_k:\zo^{\inlen(n)}\to\zo^{\outlen(n)}}_{n\in\NN,k\in\zo^{\keylen(n)}}$$ is an efficiently computable oracle-aided function family, and $R^{(\cdot),(\cdot)}$ is a \ppt oracle-aided algorithm such that the following holds. There exists a polynomial $q:\NN\to\NN$ such that for every $G=\sset{G_n:\zon\to\zo^{n+\stretch(n)}}_{n\in\NN}$ and every oracle-aided algorithm $\Break^{(\cdot)}$ for which there exists a constant $c\in\NN$ such that
	$$\abs{\ppr{k\from\zo^{\keylen(n)}}{\Break^{f^G_k}(1^n)=1}-\ppr{f\from\funcSpace_{\inlen(n),\outlen(n)}}{\Break^{f}(1^n)=1}}\geq\frac{1}{n^c}$$
	for infinitely many $n$'s, 
	$$\abs{\ppr{s\from\zon}{R^{\Break,G}(1^n,G(s))=1}-\ppr{y\from\zo^{n+\stretch(n)}}{R^{\Break,G}(1^n,y)=1}}\geq\frac{1}{q(n^c)}$$
	for infinitely many $n$'s.
	
	$(\prf,R)$ is fully black-box construction of a \emph{weak PRF}, if the above holds for all algorithms $\adv$ that query i.i.d.\ uniform queries to $f$.
\end{definition}

We next define the query complexity of a black-box construction.
\begin{definition}[$\numquery$-call adaptive and non-adaptive constructions]
	Let $\inlen(\secParam)$, $\outlen(\secParam)$, $\stretch(\secParam)$, and $\keylen(\secParam)$ be polynomials, and let 
	$$\prf=\set{f^{(\cdot)}_k:\zo^{\inlen(n)}\to\zo^{\outlen(n)}}_{n\in\NN,k\in\zo^{\keylen(n)}}$$
	be a black-box PRF construction from an $\stretch(n)$-bit stretch PRG. For a function $\numquery:\NN\to\NN$, we say that $\prf$ is a \emph{$\numquery(\cdot)$-call construction} if for every $i\in[\numquery(n)]$ there exist $S_i:\zo^{\inlen(n)}\times\zo^{\keylen(n)}\times\sparen{\zo^{n+\stretch(n)}}^{i-1}\to\zon$ and $P:\zo^{\inlen(n)}\times\zo^{\keylen(n)}\times\sparen{\zo^{n+\stretch(n)}}^{\numquery(n)}\to\zo^{\outlen(n)}$ such that
	$$f^G_k(x)=P\of{x,k,G\of{s_1},\ldots,G(s_\numquery)}$$
	for all $G:\zon\to\zo^{n+\stretch(n)}$, where $s_i=S_i(x,k,G\of{s_1},\ldots,G(s_{i-1}))$ for all $i\in[\numquery(n)]$.
	We say that the construction is \emph{non-adaptive} if every $S_i$ depends only on $x$ and $k$.
\end{definition}

In our proof, we consider a restricted notion of black-box reductions that have access to a distinguisher with oracle access to a function. Roughly, this class captures reductions where the number of times it can query the adversary is bounded by a polynomial for all sufficiently good adversaries.
\begin{definition}[Query-bounded reductions]\label{def:bounded}
	Let $(F,R)$ be a fully black-box construction of a PRF from an $\stretch(n)$-bit stretch PRG, where $$\prf=\set{f^{(\cdot)}_k:\zo^{\inlen(n)}\to\zo^{\outlen(n)}}_{n\in\NN,k\in\zo^{\keylen(n)}},$$ 
    let $\eps:\NN\to[0,1]$, and let $\numcallsdist:\NN\to\NN$. We say that $R^{G,\adv}$ is \emph{$(\numcallsdist,\eps)$-query-bounded} if for every oracle-aided algorithm $\adv^{\cdot}$ such that
    $$\abs{\ppr{k\from\zo^{\keylen(n)}}{\Break^{f^G_k}(1^n)=1}-\ppr{f\from\funcSpace_{\inlen(n),\outlen(n)}}{\Break^{f}(1^n)=1}}\geq1-\eps(n)$$
	for infinitely many $n$'s, $R^{G,\adv}$ calls $\adv$ at most $\numcallsdist(n)$ times.
\end{definition}
Note that any reduction that is $(d,\eps)$-query-bounded is also $(d,\eps')$-query-bounded for any $\eps'$ such that $\eps'(n)\geq\eps(n)$. Therefore, for our impossibility result, proving it for a \emph{larger} $\eps$ yields a stronger result. 

\section{Lower Bounds for Fully Black-Box Constructions of PRFs From PRGs}\label{sec:fullyBB}

In this section, we prove our main result. We show that there is no fully black-box construction of PRF with input length $\inlen$ from a PRG with $c$ calls to the PRG and using a query-bounded reduction, where $c=o(n/\log n)$ and $c=o(\inlen/\log\inlen)$. 

We are now ready to state and prove our main result.
\begin{theorem}\label{thm:no_fbb}
	Fix polynomials $\stretch:\NN\to\NN$ and $\inlen:\NN\to\NN$ such that $\inlen(n)=\omega(\log n)$. Then for every $c:\NN\to\NN$ such that $c(n)=o\paren{n/\log n}$ and $c(n)=o\paren{\inlen/\log\inlen}$, and for every polynomial $d:\NN\to\NN$, there is no $c(n)$-call non-adaptive fully black-box construction of a weak PRF with input length $\inlen(n)$ and a single-bit output from an $\stretch(n)$-bit stretch PRG, with a $(d,e^{-n})$-query-bounded reduction. 
\end{theorem}
\begin{proof}

We first give an overview of the proof. Our proof has 4 major steps.
\begin{description}
    \item[Step 1:] We fix $\numinputquery$ inputs $x_1,\ldots,x_\numinputquery$ for a sufficiently large $\numinputquery$. Then, for every key $k$, we identify a large subset of the inputs $\cI_k$ based on repetitions in the seeds generated by the key and the inputs in $\cI_k$. Specifically, we show that either there are $c$ seeds that are common to all inputs in $\cI_k$, or there $t<c$ common seeds and the rest of the seeds are mutually disjoint. Importantly, this step is independent of the PRG, \ie the set $\cI_k$ is the same for every PRG. We formalize the properties of the set $\cI_k$ and show its existence in \cref{lem:recursion}.

    \item[Step 2:] Next, we define an ``ideal adversary'' $\adv_\ideal^{f,G}$ that breaks the PRF by inverting the PRG $G$. This adversary is not optimal. Instead of outputting `PRF' when all input-output pairs (given by the oracle $f$) are consistent \wrt some key $k$, we define a score function based on the number of consistent input-output pairs in the set $\cI_k$, and let $\adv_\ideal^{f,G}$ output `PRF' if and only if the score is larger then a random threshold sampled uniformly at random. We prove, nonetheless, that this ideal adversary breaks the PRF. This is formalized in \cref{clm:ideal_works}.

    \item[Step 3:] We then define a ``real adversary'' whose goal is to simulate the ideal adversary without inverting the PRG. It is given the query-answer pairs generated by the reduction's interaction with the PRG. Roughly, the difference between the real adversary and the ideal adversary is in the definition of the score functions they compute. We show that the real adversary behaves close to the ideal adversary. We formally prove this in \cref{lem:close_scores_general}.

    \item[Step 4:] Finally, we combine the previous steps to argue that the reduction can successfully break the PRG using the real adversary. Since the real adversary does not query the PRG, we conclude the existence of an algorithm for breaking the PRG using a sub-exponential number of queries. Thus, instantiating the PRG with a random function will contradict \cref{lem:RO_is_PRG}, concluding the proof.
\end{description}
	Assume towards contradiction that the claim is false. That is, there exists polynomials $\inlen=\inlen(n)$, $\keylen=\keylen(n)$, and $\stretch=\stretch(n)$, there exists $c=c(n)$ such that $c=\frac{n}{w(n)\log n}$ and $c=\frac{\inlen}{w(n)\log\inlen}$ for some $w(n)=\omega(1)$, there exists a polynomial $d=d(n)$, and there exists a fully black-box construction $(\prf,R)$ such that $R$ is $(d,e^{-n})$-query-bounded. We write
  $$\prf=\set{f^{(\cdot)}_k:\zo^{\inlen}\to\zo}_{n\in\NN,k\in\zo^{\keylen}}$$
	where $f^{(\cdot)}_k$ is given by
	$$f^{G}_k(x)=P\of{x,k,G(S_1(x,k)),G(S_2(x,k)),\ldots,G(S_c(x,k))},$$
	for efficiently computable functions $S_1,S_2,\ldots,S_c:\zo^{\inlen}\times\zo^{\keylen}\to\zon$ and $P:\zo^{\inlen}\times\zo^{\keylen}\times(\zo^{n+\stretch})^c\to\zo$. 
    
	The proof follows the meta-reduction paradigm. We first define an ``ideal'' adversary that breaks the PRF using an \emph{exponential number} of queries to the PRG. We then construct an \emph{inefficient} ``real'' adversary with \emph{no} access to the PRG that is able to simulate the ideal adversary, when provided with the collection of query–answer pairs generated by the reduction’s interaction with the PRG. We then conclude that there exists an algorithm that breaks the security of the PRG using fewer than $2^{n/4}$ queries, by running the reduction and using the real adversary. 
	
	{For the argument to follow through, we have to ensure that the reduction cannot distinguish between the ideal and real adversaries. We are unable to show that the statistical distance between the responses of the adversaries is negligible. Instead, we show that the distance is $1/\poly(n)$, for some polynomial $\poly(n)$ that depends on both the number of calls $\numcallsdist$ made by the reduction to the adversary, and the number of queries $\numinputquery(n)$ made by the adversary (to its oracle function). Thus, $\poly(n)$ must be sufficiently large, implying that $\numinputquery(n)$ must be sufficiently large compared to $\numcallsdist$. Consequently, $\numcallsdist$ must be fixed prior to $\numinputquery(n)$; otherwise, the argument would suffer from a circular dependence among $\numcallsdist$ and $\numinputquery(n)$. This requirement is exactly the reason we assume the reduction to be query-bounded, as it ensures that $\numcallsdist$ is determined independently of the adversary. We note that if one could construct real and ideal adversaries whose response distributions have negligible statistical distance, our assumption would no longer be necessary.}

	We first introduce some notation. For every $t\in\set{0,1,\ldots,c}$, let $\numinputquery_t=\numinputquery_t(n)$ be a parameter (not necessarily a polynomial) to be defined by the analysis below.
	Further, let $\numinputquery=\numinputquery_0$ and fix $\numinputquery$ distinct inputs $x_1,\ldots,x_{\numinputquery}\in\zo^{\inlen}$.\footnote{Note that if we obtain a contradiction for all distinct inputs, then the same holds for random inputs (up to adding the birthday bound).}

\paragraph{Step 1: Identifying frequent seeds.}
    In the first step of the proof, our goal is to identify for every key $k\in\keyDom$ the set $\cI_k$ of inputs with many common seeds.    
	Towards defining $\cI_k$, we use the following claim, which states that for every set of $t<c$ seeds $\cS$, if every PRG seed outside of $\cS$ appears a few times among a subset $\cI_0\su[\numinputquery]$ of the inputs, then for a large fraction of these inputs, the only common seeds belong to $\cS$.
	\begin{claim}\label{clm:large_ind_set}
		Fix a number $p\in[m]$ and a key $k\in\keyDom$. Further fix a set of inputs $\cI_0\su[\numinputquery]$ and a set of seeds $\cS\su\zon$ of size $|\cS|=t\in\set{0,1,\ldots,c-1},$ such that for every $s\in\cS$ and $i\in\cI_0$ there exists $j\in[c]$ such that $s=S_j(x_i,k)$.
		Assume that there is no $s\in\zon\setminus\cS$ and $\cI\su\cI_0$ of size $p$ such that
		$$s\in\bigcap_{i\in\cI}\set{S_j(x_i,k):j\in[c]}.$$
		Then there exists a set $\cI'\su\cI_0$ of size 
		$$|\cI'|\geq \frac{|\cI_0|}{(c-t)(p-1)+1}\geq\frac{ |\cI_0|}{(c-t)p},$$
		such that the sets of PRG queries generated by them (outside of $\cS$) are mutually disjoint, \ie
		$$S_j(x_i,k)\notin\set{S_{j'}(x_{i'},k):j'\in[c],i'\in\cI'\setminus\set{i}}\setminus\cS,$$
		for every $i\in\cI'$ and $j\in[c]$.
	\end{claim}
	\begin{proof}
		Consider the graph $H=(\cI_0,E)$ where $\set{i,i'}\in E$ for $i \neq i'$ if and only if $x_i$ and $x_{i'}$ share a common PRG query outside of $\cS$, \ie
		$$\paren{\set{S_j(x_i,k):j\in[c]}\setminus\cS}\cap\paren{\set{S_j(x_{i'},k):j\in[c]}\setminus\cS}\ne\emptyset.$$
		Observe that the set $\cI'$ we aim to show exists, corresponds to an independent set in $H$. By assumption, every PRG query $s\in\zon\setminus\cS$ appears at most $p-1$ among the set of all seeds in $\set{S_j(x_i,k):i\in\cI_0,j\in[c]}$. Thus, for every vertex $i\in\cI_0$, every $j\in[c]$ such that $S_j(x_i,k)\notin\cS$ contributes 1 to the degree of $i$. Since by assumption there are exactly $t$ indexes $j\in[c]$ such that $S_j(x_i,k)\in\cS$, the degree of each vertex in $H$ is at most $(c-t)(p-1)$. Therefore, by greedily picking vertices for an independent set and removing their neighbors, we conclude that there exists an independent set of vertices of size at least $|\cI_0|/((c-t)(p-1)+1)$.
	\end{proof}

	Recall that to construct the two adversaries, we define two \emph{score functions} (one for each adversary) that aim to compute for how many inputs the reduction returned the correct output (with respect to some key). 
    The two score functions are defined \wrt each key, via a procedure that depends only on the construction of the PRF (in particular, it is independent of the output of the PRG). This is important, since the score function does not require the real adversary to make any PRG queries.

    At a high level, the procedure works in iterations, where in each iteration it looks for a seed $s$ that appears for sufficiently many inputs (\ie for many of the $x_i$'s, at least one of the seeds needed to compute $f^G_k(x_i)$ is $s$). 
    The process continues until either no $s$ appear many times, or until the process finds $c$ such seeds. In the former case, we can apply \cref{clm:large_ind_set} to argue there is a large set of inputs with mutually disjoint set of seeds, while in the latter case, we will later replace all calls to the PRG with fixed values. The following lemma summarizes this.
	
	\begin{lemma}\label{lem:recursion}
		For every key $k\in\keyDom$, there exists a number $\recSteps(k)\in\set{0,1,\ldots,c}$, there exist $\recSteps(k)$ distinct seeds $\vs_k=(s_{k,1},\ldots,s_{k,\recSteps(k)})\in(\zon)^{\recSteps(k)}$, and there exists a set $\lastset_k\su[\numinputquery]$ such that the following holds.
		\begin{enumerate}
			\item If $\recSteps(k)=c$, then $|\lastset_k|=\numinputquery_c$ and $s_{k,t}\in\bigcap_{i\in\lastset_k}\sset{S_j(x_i,k):j\in[c]}$ for every $t\in[c]$.
			
			\item Otherwise, if $\recSteps(k)<c$, then $|\lastset_k|=\frac{\numinputquery_{\recSteps(k)}}{(c-\recSteps(k)+1)\numinputquery_{\recSteps(k)+1}}$, $$S_j(x_i,k)\notin\set{S_{j'}(x_{i'},k):j'\in[c],i'\in\lastset_k\setminus\set{i}},$$
			for every $i\in\lastset_k$ and $j\in[c]$, and there exists a set $\cI\su[\numinputquery]$ of size $|\cI|=\numinputquery_{\recSteps(k)}$ such that $\lastset_k\su\cI$ and $s_{k,t}\in\bigcap_{i\in\cI}\sset{S_j(x_i,k):j\in[c]}$ for every $t\in[\recSteps(k)]$.
		\end{enumerate}
	\end{lemma}
	\begin{proof}
		The proof follows from greedily looking for the next seed $s_{k,t}$ and the (sufficiently large) set of inputs generating it, until we found $c$ such seeds, or it does not exist and we can apply \cref{clm:large_ind_set}. Formally, the lemma follows from considering the following procedure:
		\begin{enumerate}
			\item Let $\cI_{k,0}=[\numinputquery]$ and $t=1$.
			
			\item\label{step:recursion} While $t\leq c$ do the following:
			\begin{itemize}
				\item If there exists a seed $s\in\zon\setminus \set{s_{k,1},\ldots,s_{k,t-1}}$ and $\cI\su\cI_{k,t-1}$ of size $|\cI|=\numinputquery_t$ such that $s\in\bigcap_{i\in\cI}\set{S_j(x_i,k):j\in[c]}$, then let $s_{k,t}$ and $\cI_{k,t}$ denote the lexicographically smallest such seed $s$ and set $\cI$, and increment $t$ by 1.
				
				\item Otherwise, by \cref{clm:large_ind_set}, there exists a set of inputs $\cI\su\cI_{k,t-1}$ of size
				\ifdefined\IsLLNCS
				$$|\cI|=\frac{|\cI_{k,t-1}|}{(c-t)\cdot \numinputquery_t}=\frac{\numinputquery_{t-1}}{(c-t)\cdot \numinputquery_t}$$
				\else 
				$$|\cI|=\frac{|\cI_{k,t-1}|}{(c-t)\cdot \numinputquery_t}=\frac{\numinputquery_{t-1}}{(c-t)\cdot \numinputquery_t}$$
				\fi
				such that
				$$S_j(x_i,k)\notin\set{S_{j'}(x_{i'},k):j'\in[c],i'\in\cI\setminus\set{i}},$$
				for every $i\in\cI$ and $j\in[c]$. Let $\recSteps(k)=t-1$, and let $\lastset_k$ denote lexicographically smallest such set $\cI$, and halt.
			\end{itemize}
			
			\item Let $\recSteps(k)=c$ and $\lastset_k=\cI_{k,c}$.
		\end{enumerate}
	\end{proof}
    It will be convenient to consider the value $\nu=\min_k |\lastset_k|$.

    \paragraph{Step 2: Defining the ideal adversary.}
	We next define the first score function, which we call the \emph{ideal score function}. It depends on the entire PRG oracle and will only be used by the ideal adversary. Given a key and a sequence of possible outputs $\vz\in(\zo)^\numinputquery$,\footnote{The score also implicitly depends on the inputs $x_1,\ldots,x_\numinputquery$ we fixed earlier.} the score is defined by using \cref{lem:recursion} to find the seeds that appear many times. It then replaces the output of $G$ on those seeds with fixed values, and returns the number of inputs in $\lastset_k$ that are consistent with $\vz$.
	
	In what follows, for every pair of sequences $\vs=(s_1,\ldots,s_t)\in(\zon)^t$ and $\vy=(y_1,\ldots,y_t)\in(\zo^{n+\stretch})^t$, define
	$$\prf_{\vs,\vy}=\set{f^{(\cdot)}_{k,\vs,\vy}:\zo^{\inlen}\to\zo}_{n\in\NN,k\in\zo^{\keylen}},$$
	where $f_{k,\vs,\vy}^{(\cdot)}$ is the same as $f^{(\cdot)}_{k}$, 
	except that for oracle queries in $\vs$,
	the oracle answer for the seed $s_i$ is replaced with $y_i$, for every $i\in[t]$.
	
	\begin{definition}[Ideal score function]
		For every key $k\in\keyDom$, every sequence of possible outputs $\vz=(z_1,\ldots,z_\numinputquery)$, and every function $G:\zon\to\zo^{n+\stretch}$, define the ideal score of $k$ to be
		$$\iscore_{k}^G(\vz)=\max_{\vy\in(\zo^{n+\stretch})^{\recSteps(k)}}\of{\abs{\set{i\in\lastset_k:f^G_{k,\vs_k,\vy}(x_i)=z_i}}}.$$
	\end{definition}
	Observe that if $\recSteps(k)=c$, then $\iscore^G_k(\vz)$ is independent of $G$.
	
	Next, we define the ``ideal'' adversary $\Break^{f,G}_{\ideal}$ for the PRF. This adversary queries the PRG on \emph{exponentially many} inputs. We will then show that the reduction can simulate this adversary efficiently, resulting in an efficient distinguisher for the PRG. Roughly, $\Break^{f,G}_{\ideal}$ works as follows. First, it queries the oracle $f$ on $x_1,\ldots,x_{\numinputquery}$. The adversary outputs PRF if and only if there is a key whose ideal score is at least $\alpha\cdot|\lastset_k|$, where $\alpha\in[0.8,1]$ \emph{is sampled uniformly at random}. 
	
	\begin{algorithm}[$\Break_{\ideal}$]~\\
		{Input:} The security parameter $1^n$.\\
		{Oracle access:} A function $f:\zo^{\inlen}\to\zo$ and a PRG $G:\zon\to\zo^{n+\stretch}$.
		
		\begin{enumerate}
			\item Query $G$ on every input in $\zon$.
			\item Query $f$ on $x_1,\ldots,x_\numinputquery$. Let $\vz=(z_1,\ldots,z_\numinputquery)$ denote the respective answers.
			\item Sample $\alpha\in[0.8,1]$ uniformly at random.
			\item Output 1 if and only if there exists $k\in\keyDom$ such that $\iscore^{G}_k(\vz)\geq\alpha\cdot|\lastset_k|$.
		\end{enumerate}
	\end{algorithm}	
	
	We first show that $\Break^{f,G}_{\ideal}$ indeed distinguishes the PRF from a random function.
	
	\begin{claim}\label{clm:ideal_works}
		For all $G:\zon\to\zo^{n+\stretch}$,
		\begin{align*}
			\abs{\ppr{k\from\zo^{\keylen}}{\Break_{\ideal}^{f^G_k,G}(1^n)=1}-\ppr{f\from\funcSpace_\inlen}{\Break_{\ideal}^{f,G}(1^n)=1}}\geq1-e^{\keylen+c(n+\stretch)-\nu/10}.
		\end{align*}
	\end{claim}
	\begin{proof}
		We first show that $\Break_{\ideal}^{f^G_k,G}(1^n)$ always outputs 1. Indeed, observe that
		$$\iscore_{k}^G(\vz)=\max_{\vy\in(\zo^{n+\stretch})^{\recSteps(k)}}\of{\abs{\set{i\in\lastset_k:f^G_{k,\vs_k,\vy}(x_i)=f_k^G(x_i)}}}.$$
		Then taking $y_i=G(s_{k,i})$ for every $i\in[\recSteps(k)]$ results in $\iscore_{k}^G(\vz)=|\lastset_k|$.
		
		To conclude the proof, we next show that for a random function $f$, $\Break_{\ideal}^{f,G}(1^n)$ outputs 1 with negligible probability. For every key $k\in\keyDom$, every $\vy\in(\zo^{n+\stretch})^{\recSteps(k)}$, and every $i\in\lastset_k$, define the random variable $X_{k,\vy,i}$ to be the indicator for the event $f_{k,\vs_k,\vy}^G(x_i)=f(x_i)$. Then
		$$\iscore_{k}^G(\vz)=\max_{\vy\in(\zo^{n+\stretch})^{\recSteps(k)}}\of{\sum_{i\in\lastset_k}X_{k,\vy,i}}.$$
		Since $f$ is a random function, these indicators are independent and have expectation 1/2. Therefore, for every $k\in\keyDom$ and every $\vy\in(\zo^{n+\stretch})^{\recSteps(k)}$, Hoeffding's inequality (see \cref{thm:Hoeffding}) implies that
		\ifdefined\IsLLNCS
		\begin{align*}
			\ppr{f\from\funcSpace_\inlen}{\sum_{i\in\lastset_k}X_{k,\vy,i}>0.8\cdot\abs{\lastset_k}}
			&\leq e^{-2\cdot(0.8-0.5)^2\cdot\abs{\lastset_k}}\\
			&\leq e^{-\abs{\lastset_k}/10}.
		\end{align*}
		\else
		\begin{align*}
			\ppr{f\from\funcSpace_\inlen}{\sum_{i\in\lastset_k}X_{k,\vy,i}>0.8\cdot\abs{\lastset_k}}\leq e^{-2\cdot(0.8-0.5)^2\cdot\abs{\lastset_k}}\leq e^{-\abs{\lastset_k}/10}\leq e^{-\nu/10}.
		\end{align*}
		\fi
		Therefore, by the union bound,
		\ifdefined\IsLLNCS
		\begin{align*}
			&\ppr{f\from\funcSpace_\inlen}{\Break^{f,G}_{\ideal}(1^n)=1}\\
			&\quad\leq\ppr{f\from\funcSpace_\inlen}{\exists k\in\keyDom:\iscore^G_k(f(x_1),\ldots,f(x_\numinputquery))\geq0.8\cdot\abs{\lastset_k}}\\
			&\quad=\ppr{f\from\funcSpace_\inlen}{\exists k\in\keyDom\;\exists\vy\in(\zo^{n+\stretch})^{\recSteps(k)}:\sum_{i\in\lastset_k}X_{k,\vy,i}\geq0.8\cdot\abs{\lastset_k}}\\
			&\quad\leq 2^{\keylen+c(n+\stretch)}\cdot e^{-\nu/10}\\
			&\quad\leq e^{\keylen+c(n+\stretch)-\nu/10}.
		\end{align*}
		\else
		\begin{align*}
			\ppr{f\from\funcSpace_\inlen}{\Break^{f,G}_{\ideal}(1^n)=1}
			&\leq\ppr{f\from\funcSpace_\inlen}{\exists k\in\keyDom:\iscore^G_k(f(x_1),\ldots,f(x_\numinputquery))\geq0.8\cdot\abs{\lastset_k}}\\
			&=\ppr{f\from\funcSpace_\inlen}{\exists k\in\keyDom\;\exists\vy\in(\zo^{n+\stretch})^{\recSteps(k)}:\sum_{i\in\lastset_k}X_{k,\vy,i}\geq0.8\cdot\abs{\lastset_k}}\\
			&\leq 2^{\keylen+c(n+\stretch)}\cdot e^{-\nu/10}\\
			&\leq e^{\keylen+c(n+\stretch)-\nu/10}.
		\end{align*}
		\fi
	\end{proof}
	Recall that for every $k\in\keyDom$, $|\lastset_k|=\numinputquery_c$ or $|\lastset_k|=\frac{\numinputquery_{t-1}}{(c-t+1)\numinputquery_t}$ for some $t\in[c]$. Since $\nu\geq|\lastset_k|$ for every $k\in\keylen$, taking $\numinputquery_c\geq 10(\keylen+c(n+\stretch)+n)$ and $\numinputquery_{t-1}\geq (c-t+1)\numinputquery_t\cdot 10(\keylen+c(n+\stretch)+n)$ for all $t\in[c]$, results in the ideal adversary's advantage begin at least $1-e^{-n}$.

    \paragraph{Step 3: Defining the real adversary.}
	We now define the \emph{real score}, which will be used by the real adversary. It is defined \wrt a set of query-answer pairs $\cG$ instead of a PRG. The set $\cG$ will later be instantiated \wrt the queries the reduction made to the PRG.
    Similarly to the ideal score, the computation of the real score also starts by replacing calls to $G$ on seeds that appear many times with fixed inputs. The only difference is the end of the recursive process. There are two cases: either every call to $G$ was replaced, or the inputs in $\lastset_k$ have mutually disjoint seeds. In the former case, the real score is identical to the ideal score. In the latter case, the score is defined as follows. Every input $x_i$ for which every remaining seed (\ie for which the output of $G$ was not fixed) belongs to the set $\cG$ can be computed using $G$. Such inputs contribute the same value as in the ideal score. For other inputs, the real score estimates its contribution to the ideal score by computing the probability over the PRG that the answer is correct. Since this is the case where all remaining seeds are mutually disjoint, the outputs corresponding to these inputs are independent random variables, and we can estimate their contribution to the score with high precision (see \cref{lem:close_scores_general} below).

    In the following, for a set of PRG query-answer pairs $\cG=\sset{(s^\cG_1,y^\cG_1),\ldots,(s^\cG_q,y^\cG_q)}\su\zon\times\zo^{n+\stretch}$ we let $\vs^\cG$ denote the sequence of queries in $\cG$ that \emph{do not appear in $\vs_k$}, and let $\vy^\cG$ denote the corresponding sequence of outputs. The reason not to include the queries in $\vs_k$ is that any query $s$ made by the reduction is unknown to the ideal adversary, which maximizes its score over all possible values of $G(s)$. Thus, the real adversary does the same, obtaining consistency between the two scores.
	\begin{definition}[Real score function]
		For every key $k\in\keyDom$, every sequence of possible outputs $\vz=(z_1,\ldots,z_\numinputquery)$, and every set of PRG query-answer pairs $\cG=\sset{(s_1^\cG,y_1^\cG),\ldots,(s_q^\cG,y_q^\cG)}\in\zon\times\zo^{n+\stretch}$, define $\rscore_{k}^\cG(\vz)$ as follows. 

		\begin{itemize}
			\item If $\recSteps(k)=c$, then 
			let
			$$\rscore_{k}^\cG(\vz)=\max_{\vy\in(\zo^{n+\stretch})^{c}}\of{\abs{\set{i\in\lastset_k:P\of{x_i,k,y'_1,\ldots,y'_c}=z_i}}},$$
			where for every $j\in[c]$, $y'_j=y_{j'}$ for the unique $j'\in[c]$ such that $S_j(x_i,k)=s_{k,j'}$.
			
			\item 
            If $\recSteps(k)<c$, then let
			\begin{align*}
				\rscore_{k}^\cG(\vz)				=\max_{\vy\in(\zo^{n+\stretch})^{\recSteps(k)}}\of{\sum_{i\in\lastset_k}\ppr{G'\from\funcSpace_{n,n+\stretch}}{f^{G'}_{k,(\vs_k\conc \vs^\cG),(\vy\conc\vy^\cG)}(x_i)=z_i}}.
			\end{align*}
		\end{itemize}
	\end{definition}
    
	Similarly to the ideal score, if $\recSteps(k)=c$ then $\rscore_{k}^\cG(\vz)$ is independent of $\cG$.
    In particular, in such a case, $\rscore_{k}^\cG\of{\vz}=\iscore_{k}^G\of{\vz}$.

	Since every ideal adversary we consider has advantage at least $1-e^{-n}$, it follows that there exists a constant $\gamma\in\NN$ such that for any $G\in\funcSpace_{n,n+\stretch}$,
	\begin{align}\label{eq:R_dist}
		\abs{\ppr{s\from\zon}{R^{\Break_{\ideal},G}(1^n,G(s))=1}-\ppr{y\from\zo^{n+\stretch}}{R^{\Break_{\ideal},G}(1^n,y)=1}}\geq\frac{1}{n^\gamma}
	\end{align}
	for infinitely many $n$'s. To obtain a contradiction, we construct an algorithm $\tilR^{(\cdot)}$ that successfully breaks a \emph{randomly chosen} PRG with noticeable probability, using at most $2^{n/4}$ queries to $G$ and \emph{without} having access to $\Break_\ideal$. Specifically, we define $\tilR$ such that
	\ifdefined\IsLLNCS
	\begin{align}\label{eq:tilR_dist}
		\Pr_{G\from\funcSpace_{n,n+\stretch}}
		\left[\Adv\of{\tilR^G}\leq\frac{1}{3n^{\gamma}}\right]\leq\frac34,
	\end{align}
	\else
	\begin{align}\label{eq:tilR_dist}
		\ppr{G\from\funcSpace_{n,n+\stretch}}{\abs{\ppr{s\from\zon}{\tilR^{G}(1^n,G(s))=1}-\ppr{y\from\zo^{n+\stretch}}{\tilR^{G}(1^n,y)=1}}\leq\frac{1}{3n^{\gamma}}}\leq\frac34,
	\end{align}
	\fi
	 \ifdefined\IsLLNCS for infinitely many $n$'s, where 
	 $$\Adv\of{\tilR^G}:=\abs{\ppr{s\from\zon}{\tilR^{G}(1^n,G(s))=1}-\ppr{y\from\zo^{n+\stretch}}{\tilR^{G}(1^n,y)=1}},$$
	 is the advantage of $\tilR^G$.
	 \else for infinitely many $n$'s. \fi Note that \cref{eq:tilR_dist} contradicts \cref{lem:RO_is_PRG}.
	
	Before defining $\tilR$, we first introduce an adversary $\adv_{\real}$ that, intuitively, efficiently emulates $\Break_{\ideal}$ using the set of queries to $G$ used so far by the reduction. Roughly, $\adv_{\real}$ works similarly to $\Break_{\ideal}$, but uses $\rscore$ instead of $\iscore$. We then show by letting $\tilR$ run $R$, and replacing each call to $\Break_\ideal$ with running $\adv_{\real}$ using the set of query-answer pairs of the PRG used so far, the answers of the two distinguishing oracles are close in statistical distance.
	
	\begin{algorithm}[$\adv_{\real}$]~\\
		{Input:} The security parameter $1^n$, a sequence $\vz\in\zo^{\numinputquery}$, and a set of query-answer pairs $\cG\su\zon\times\zo^{n+\stretch}$.
		
		\begin{enumerate}
			\item Sample $\alpha\in[0.8,1]$ uniformly at random.
			\item Output 1 if and only if there exists $k\in\keyDom$ such that $\rscore^\cG_k(\vz)\geq\alpha\cdot|\lastset_k|$.
		\end{enumerate}
	\end{algorithm}
	
	We are now ready to define $\tilR$. The algorithm simply runs $R$, replacing all calls to $\Break_{\ideal}$ by running $\adv_{\real}$.
	
	\begin{algorithm}[$\tilR$]~\\
		{Input:} The security parameter $1^n$ and a challenge $\chall\in\zo^{n+\stretch}$.\\
		{Oracle access:} A PRG $G:\zon\to\zo^{n+\stretch}$.
		
		\begin{enumerate}
			\item Simulate $R^{\Break,G}$ as follows.
			\begin{itemize}
				\item Whenever $R$ calls $G$, query $G$ with the same input.
				
				\item Whenever $R$ calls $\Break_{\ideal}$, send it $x_1,\ldots,x_\numinputquery$ (if $R$ queries $G$ during this call then query it as well). Let $\vz=(z_1,\ldots,z_\numinputquery)\in\zo^\numinputquery$ denote the respective answers of $R$. Call $\adv_{\real}(1^n,\vz,\cG)$, where $\cG$ is the current set of query-answer pairs \wrt $G$, and send to $R$ the same answer.
			\end{itemize}
			
			\item Output whatever $R$ outputs.
		\end{enumerate}
	\end{algorithm}
	
	We next show that $\tilR$ and $R$ output the same value with high probability. {We do so by showing that the statistical distance between the states of $R$ and $\tilR$ after every call to the adversary is small. Towards proving this, we first show that for every possible sequence of responses $\vz$ made by the reduction, the two scores computed by the real and ideal adversaries at any given call are somewhat close except with a negligible probability, where the probability is taken over $G$ and $\chall$ (the randomness of the reduction is fixed).} In the following, for a set $\cG\su\zon\times\zo^{n+\stretch}$ we let $\funcSpace_{n,n+\stretch}^\cG$ denote the set of all functions that are consistent with $\cG$. 
	
	\begin{lemma}\label{lem:close_scores_general}
		Consider a call made by the reduction to the distinguishing oracle, and fix the query-answer pairs $\cG\su\zon\times\zo^{n+\stretch}$ made so far by the reduction. Let $Y$ denote the distribution over $\zo^{n+\stretch}$ for the challenge given to the reduction (\ie it is either uniform over $\zo^{n+\stretch}$ or is $G(S)$ for $S$ uniform over $\zon$), and let $\vz\in(\zo)^{\numinputquery}$ denote the reduction's answers to the distinguishing oracle on inputs $x_1,\ldots,x_\numinputquery$. Then for every $\delta\geq2/\nu$,

		\begin{align}	
        \label{eq:close_scores0}
        \ppr{G\from\funcSpace_{n,n+\stretch}^\cG,\chall\from Y}{\exists k\in\keyDom:\abs{\iscore^{G}_k(\vz)-\rscore^\cG_{k}(\vz)}\geq \delta\cdot\abs{\lastset_k}}\leq e^{\keylen-\frac{\delta^2}{2}\cdot\nu}. 
        \end{align}
	\end{lemma}
	\begin{proof}
		By the union bound, it suffices to show that 
		\begin{align}\label{eq:close_scores}
			\ppr{G\from\funcSpace_{n,n+\stretch}^\cG,\chall\from Y}{\abs{\iscore^{G}_k(\vz)-\rscore_{k}^\cG(\vz)}\geq \delta\cdot|\lastset_k|}\leq e^{-\frac{\delta^2}{2}\cdot\nu} 
		\end{align}
		for every key $k\in\keyDom$.    
        
        As noted before, if $\recSteps(k)=c$, then the two scores are identical and the difference is 0. In the remainder of the proof, we assume that $\recSteps(k)<c$.
        
        For every $\vy\in(\zo^{n+\stretch})^{\recSteps(k)}$, let
        $$\iscore_{k}^G(\vz,\vy)=\abs{\set{i\in\lastset_k:f^G_{k,\vs_k,\vy}(x_i)=z_i}}$$
        and let
        $$\rscore_{k}^\cG(\vz,\vy)=\sum_{i\in\lastset_k}\ppr{G'\from\funcSpace_{n,n+\stretch}}{f^{G'}_{k,(\vs_k\conc \vs^\cG),(\vy\conc\vy^\cG)}(x_i)=z_i},$$
        where $\vs^\cG$ is the sequence of seeds in $\cG$ that do not appear in $\vs_k$, and $\vy^\cG$ is the sequence of corresponding outputs.
        
		We show that for any sequence $\vy\in(\zo^{n+\stretch})^{\recSteps(k)}$, 
		\begin{align}
        \label{eq:close_scores2}
        \ppr{G\from\funcSpace_{n,n+\stretch}^\cG,\chall\from Y}{\abs{\iscore^{G}_k(\vz,\vy)-\rscore_{k}^\cG(\vz,\vy)}\geq \delta\cdot|\lastset_k|}\leq e^{-\frac{\delta^2}{2}\cdot\nu}. 
        \end{align}
		In particular, this holds for the $\vy$ that maximizes the scores, which implies \cref{eq:close_scores}.
        Note that the string $\vy$ that maximizes the scores for the real and ideal adversaries could be different, but~\cref{eq:close_scores2} implies~\cref{eq:close_scores} nevertheless.

        Fixing $k$ and $\vy\in(\zo^{n+\stretch})^{\recSteps(k)}$,
        we simplify notation by writing $\iscore$ and $\rscore$ for the ideal and real scores, respectively. 
		We next show that 
        \begin{align}\label{eq:exp_scores_general}            \rscore\leq\eex{G\from\funcSpace_{n,n+\stretch}^\cG,\chall\from Y}{\iscore}\leq\rscore+1.
		\end{align}
		By linearity of expectation, \ifdefined\IsLLNCS it follows that \else\fi the contribution of every $i\in\lastset_k$ to $\seex{G}{\iscore}$ is 
		$$\ppr{G\from\funcSpace_{n,n+\stretch}^\cG,\chall\from Y}{f^{G}_{k,\vs_k,\vy}(x_i)=z_i}$$
		(where $z_i$ may depend on $\chall$ and $G$). 

        If $Y$ is the uniform distribution over $\zo^{n+\stretch}$, then 
        $$\ppr{G\from\funcSpace_{n,n+\stretch}^\cG,\chall\from Y}{f^{G}_{k,\vs_k,\vy}(x_i)=z_i}=\ppr{G'\from\funcSpace_{n,n+\stretch}^\cG,\chall\from Y}{f^{G'}_{k,(\vs_k\conc\vs^\cG),(\vy\conc\vy^\cG)}(x_i)=z_i}$$
        for every $i\in\lastset_k$, hence $\seex{G,\chall}{\iscore}=\rscore$. Otherwise, if the distribution of $Y$ is $G(S)$ where $S$ is uniform over $\zon$, then every $i\in\lastset_k$ such that $z_i=P(x_i,k,\vy')$, where $Y$ appears in $\vy'$, contributes 1 to $\seex{G,\chall}{\iscore}$. Since the sets $\sset{S_j(x_i,k):j\in[c]}$ are pairwise disjoint, there is at most 1 such $i$. On the other hand, for the real score, each $i$ contributes a value between 0 and 1, thus $\seex{G,\chall}{\iscore}\leq\rscore+1$.

		For every $i\in\lastset_k$ let $X_i$ denote the indicator for the event $f^G_{k,\vs_k,\vy}(x_i)=z_i$ (over sampling $G\from\funcSpace_{n,n+\stretch}^\cG$ and $\chall\from Y$). Then $\iscore=\sum_{i\in\lastset_k}X_i$. These random variables are independent since the sets $\sset{S_j(x_i,k):j\in[c]}$ are pairwise disjoint. Since $\delta\geq2/\nu\geq2/|\lastset_k|$, by Hoeffding's inequality, it follows that
		\ifdefined\IsLLNCS
		\begin{align*}
			&\ppr{G\from\funcSpace^\cG_{n,n+\stretch},\chall\from Y}{\abs{\iscore-\rscore}\geq \delta\cdot\abs{\lastset_k}}\\
			&\quad\leq\ppr{G\from\funcSpace^\cG_{n,n+\stretch},\chall\from Y}{\abs{\iscore-\ex{\iscore}}\geq \delta\cdot\abs{\lastset_k}-1}\\
			&\quad\leq\ppr{G\from\funcSpace^\cG_{n,n+\stretch},\chall\from Y}{\abs{\iscore-\ex{\iscore}}\geq \frac{\delta}{2}\cdot\abs{\lastset_k}}\\
			&\quad\leq e^{-\frac{2\paren{\frac{\delta}{2}\cdot|\lastset_k|}^2}{\abs{\lastset_k}}}\\
			&\quad= e^{-\frac{\delta^2}{2}\cdot\abs{\lastset_k}}.
		\end{align*}
		\else
		\begin{align*}
			\ppr{G\from\funcSpace^\cG_{n,n+\stretch},\chall\from Y}{\abs{\iscore-\rscore}\geq \delta\cdot\abs{\lastset_k}}
			&\leq\ppr{G\from\funcSpace^\cG_{n,n+\stretch},\chall\from Y}{\abs{\iscore-\ex{\iscore}}\geq \delta\cdot\abs{\lastset_k}-1}\\
			&\leq\ppr{G\from\funcSpace^\cG_{n,n+\stretch},\chall\from Y}{\abs{\iscore-\ex{\iscore}}\geq \frac{\delta}{2}\cdot\abs{\lastset_k}}\\
			&\leq e^{-\frac{2\paren{\frac{\delta}{2}\cdot|\lastset_k|}^2}{\abs{\lastset_k}}}\\
			&=e^{-\frac{\delta^2}{2}\cdot\abs{\lastset_k}}\\
			&\leq e^{-\frac{\delta^2}{2}\cdot\nu}.
		\end{align*}
		\fi
	\end{proof}

	We now show that for any call by the reduction to the distinguishing oracle, the statistical distance (taken over $G$ and the challenge $\chall$) between the answers of $\Break_{\ideal}$ and $\adv_{\real}$ is small.
	\begin{lemma}\label{lem:single_call_general}
		Fix a call to the distinguishing oracle. Let $\cG\su\zon\times\zo^{n+\stretch}$ denote the set of query-answer pairs made by $R$ to $G$, let $Y$ denote the distribution over the challenge, and let $\vz=(z_1,\ldots,z_\numinputquery)\in\zo^{\numinputquery}$ denote the answers of $R$ when receiving inputs $x_1,\ldots,x_\numinputquery$ from the distinguishing oracle. Finally, for every $G\in\funcSpace_{n,n+\stretch}^\cG$ and $\chall\in\zo^{n+\stretch}$ let $\iprob^{G,\chall}$ denote the probability that $\Break_{\ideal}(1^n)$ outputs 1 when given $\vz$ as the oracle answers, and let $\rprob^{\cG,\chall}$ denote the probability that $\adv_{\real}(1^n,\vz,\cG)$ outputs 1. Then, for every $\delta\geq2/\nu$ 
		$$\SD_{G\from\funcSpace_{n,n+\stretch}^\cG,\chall\from Y}\of{\iprob^{G,\chall},\rprob^{\cG,\chall}}\leq e^{\keylen-\frac{\delta^2}{2}\cdot\nu}+5\delta.$$
	\end{lemma}
	\begin{proof}
		For every $G\in\funcSpace_{n,n+\stretch}^{\cG}$ and every $\chall\in\zo^{n+\stretch}$ let 
		$$\mu^{G,\chall}=\max_{k\in\keyDom}\frac{\abs{\iscore^{G}_k(\vz)-\rscore^\cG_{k}(\vz)}}{\abs{\lastset_k}}.$$
		Additionally, for every $\alpha\in[0.8,1]$ let 
		$$
		\iprob^{G,\chall}(\alpha)=
		\begin{cases}
			1 & \text{if }\exists k\in\keyDom: \iscore^G_k(\vz)\geq\alpha|\lastset_k|\\
			0 & \text{otherwise},
		\end{cases}
		$$
		and let
		$$
		\rprob^{\cG,\chall}(\alpha)=
		\begin{cases}
			1 & \text{if }\exists k\in\keyDom: \rscore^\cG_k(\vz)\geq\alpha|\lastset_k|\\
			0 & \text{otherwise}.
		\end{cases}
		$$
		Then by Jensen's inequality,
		\begin{align*}
			\abs{\eex{\alpha\from[0.8,1]}{\iprob^{G,\chall}(\alpha)-\rprob^{\cG,\chall}(\alpha)}}
			\leq\eex{\alpha\from[0.8,1]}{\abs{\iprob^{G,\chall}(\alpha)-\rprob^{\cG,\chall}(\alpha)}}.
		\end{align*}
		Note that any $\alpha$ contributes to the expectation on the right-hand side if and only if $\alpha$ is between $\iscore_{k}^G(\vz)/\sabs{\lastset_k}$ and $\rscore_{k'}^\cG(\vz)/\sabs{\cI_{k',\recSteps'(k')}}$ for (possibly distinct) $k,k'\in\keyDom$. Observe that this distance is at most $\mu^{G,\chall}$ (\ie the distance \wrt the same key). Therefore,
        \ifdefined\IsLLNCS
        \begin{align*}
			&\eex{\alpha\from[0.8,1]}{\abs{\iprob^{G,\chall}(\alpha)-\rprob^{\cG,\chall}(\alpha)}}\\
			&\quad\leq\ppr{\alpha\from[0.8,1]}{\exists k\in\keyDom: \frac{\min\set{\iscore^{G}_k(\vz),\rscore^\cG_{k}(\vz)}}{\abs{\lastset_k}}\leq\alpha\leq\frac{\max\set{\iscore^{G}_k(\vz),\rscore^\cG_{k}(\vz)}}{\abs{\lastset_k}}}\\
			&\quad\leq\frac{\mu^{G,\chall}}{1-0.8}\\
			&\quad=5\mu^{G,\chall}.
		\end{align*}
        \else
		\begin{align*}
			&\eex{\alpha\from[0.8,1]}{\abs{\iprob^{G,\chall}(\alpha)-\rprob^{\cG,\chall}(\alpha)}}\\
			&\quad\leq\ppr{\alpha\from[0.8,1]}{\exists k\in\keyDom: \min\set{\iscore^{G}_k(\vz),\rscore^\cG_{k}(\vz)}\leq\alpha\cdot\abs{\lastset_k}\leq\max\set{\iscore^{G}_k(\vz),\rscore^\cG_{k}(\vz)}}\\
			&\quad\leq\frac{\mu^{G,\chall}}{1-0.8}\\
			&\quad=5\mu^{G,\chall}.
		\end{align*}
        \fi
		Hence, by \cref{lem:close_scores_general}, 
		\begin{align*}
			&\ppr{G\from\funcSpace_{n,n+\stretch}^\cG,\chall\from Y}{\abs{\eex{\alpha}{\iprob^{G,\chall}(\alpha)-\rprob^{\cG,\chall}(\alpha)}}\geq 
			5\delta}\\
			&\quad\leq\ppr{G\from\funcSpace_{n,n+\stretch}^\cG,\chall\from Y}{\exists k\in\keyDom:\abs{\iscore^{G}_k(\vz)-\rscore^\cG_{k}(\vz)}\geq 
			\delta\cdot|\lastset_k|}\\
			&\quad\leq e^{\keylen-\frac{\delta^2}{2}\cdot\nu}.
		\end{align*}
		Therefore, the statistical distance is bounded by
		\ifdefined\IsLLNCS
		\begin{align*}
			&\SD_{G\from\funcSpace_{n,n+\stretch}^\cG,\chall\from Y}\of{\iprob^{G\chall},\rprob^{\cG,\chall}}\\
			&\quad\leq\eex{G\from\funcSpace_{n,n+\stretch}^\cG,\chall\from Y}{\abs{\iprob^{G,\chall}-\rprob^{\cG,\chall}}}\\
			&\quad=\eex{G\from\funcSpace_{n,n+\stretch}^\cG,\chall\from Y}{\abs{\eex{\alpha}{\iprob^{G,\chall}(\alpha)-\rprob^{\cG,\chall}(\alpha)}}}\\
			&\quad\leq\ppr{G\from\funcSpace_{n,n+\stretch},\chall\from Y}{\abs{\eex{\alpha}{\iprob^{G,\chall}(\alpha)-\rprob^{\cG,\chall}(\alpha)}}\geq
			5\delta}+5\delta\\
			&\quad\leq e^{\keylen-\frac{\delta^2}{2}\cdot\nu}+5\delta.
		\end{align*}
		\else
		\begin{align*}
			\SD_{G\from\funcSpace_{n,n+\stretch}^\cG,\chall\from Y}\of{\iprob^{G,\chall},\rprob^{\cG,\chall}}
			&\leq\eex{G\from\funcSpace_{n,n+\stretch}^\cG,\chall\from Y}{\abs{\iprob^{G,\chall}-\rprob^{\cG,\chall}}}\\
			&=\eex{G\from\funcSpace_{n,n+\stretch}^\cG,\chall\from Y}{\abs{\eex{\alpha}{\iprob^{G,\chall}(\alpha)-\rprob^{\cG,\chall}(\alpha)}}}\\
			&\leq\ppr{G\from\funcSpace_{n,n+\stretch}^\cG,\chall\from Y}{\abs{\eex{\alpha}{\iprob^{G,\chall}(\alpha)-\rprob^{\cG,\chall}(\alpha)}}\geq
			5\delta}+5\delta\\
			&\leq e^{\keylen-\frac{\delta^2}{2}\cdot\nu}+5\delta.
		\end{align*}
		\fi
	\end{proof}

    For every $t\in\sset{0,1,\ldots,\numcallsdist}$, let $\st_t$ and $\widetilde{\st}_t$ denote the states of $R$ and $\tilR$, respectively, after the \tth call to the distinguishing oracle, where the state consists of, the challenge $\chall$, the (current) query-answer pairs $\cG$ made to the PRG, and the $t$ answers of the distinguishing oracle.\footnote{Formally, it should also depend on the randomness. Since all our claims hold for every randomness, we ignore it for simplicity.} Then by \cref{lem:single_call_general} and the triangle inequality, for every $t\in[\numcallsdist]$, after the \tth call to the distinguishing oracle, 
	$$\SD\of{\st_t,\widetilde{\st}_t}\leq\SD\of{\st_{t-1},\widetilde{\st}_{t-1}}+e^{\keylen-\frac{\delta^2}{2}\cdot\nu}+5\delta\leq t\cdot e^{\keylen-\frac{\delta^2}{2}\cdot\nu}+5t\delta,$$
	{where the statistical distances are taken over $G\from\funcSpace_{n,n+\stretch}^\cG$, $\chall\from Y$ where $Y$ is the distribution of the challenge, and over the reductions' randomness}. Thus,
	\begin{align}\label{eq:out_same}
		\pr{\tilR^G(\chall)\ne R^{\Break_{\ideal},G}(\chall)}\leq \numcallsdist\cdot e^{\keylen-\frac{\delta^2}{2}\cdot\nu}+5\numcallsdist\delta,
	\end{align}
	{where the probability is taken over $G\from\funcSpace_{n,n+\stretch}^\cG$, $\chall\from Y$ where $Y$ is the distribution of the challenge, and over the reductions' randomness}.

    \paragraph{Step 4: Putting it all together.} 
    Our goal is to upper-bound
	\ifdefined\IsLLNCS
	\begin{align*}
		\Delta:=\Pr_{G\from\funcSpace_{n,n+\stretch}}
		\left[
		\begin{aligned}
			&\left|\ppr{s\from\zon}{\tilR^{G}(1^n,G(s))=1}\right.\\
			&\qquad\left.-\ppr{y\from\zo^{n+\stretch}}{\tilR^{G}(1^n,y)=1}\right|\leq\frac{1}{3n^{\gamma}}
		\end{aligned}
		\right]
	\end{align*}
	\else
	$$\Delta:=\ppr{G\from\funcSpace_{n,n+\stretch}}{\abs{\ppr{s\from\zon}{\tilR^{G}(1^n,G(s))=1}-\ppr{y\from\zo^{n+\stretch}}{\tilR^{G}(1^n,y)=1}}\leq\frac{1}{3n^{\gamma}}}.$$
	\fi
    
	Towards bounding $\Delta$, we first upper-bound the expected difference between the answers of $R$ and $\tilR$ (where the expectation is taken over $G$). To alleviate notations, throughout the rest of the proof we remove $\Break_{\ideal}$ from the notations and write $R^G$ instead. {Additionally, all probabilities are taken over the reductions' randomness and we ignore this for brevity.}
	\begin{claim}\label{clm:exp_diff}
		Let $Y$ denote the distribution over the challenge. Then for all sufficiently large $n$'s,
		$$\eex{G\from\funcSpace_{n,n+\stretch}}{\abs{\ppr{\chall\from Y}{\tilR^G(\chall)=1}-\ppr{\chall\from Y}{R^G(\chall)=1}}}\leq\numcallsdist\cdot e^{\keylen-\frac{\delta^2}{2}\cdot\nu}+5\numcallsdist\delta.$$
	\end{claim}
	\begin{proof}
		For every $G\in\funcSpace_{n,n+\stretch}$, observe that
		$$\ppr{\chall\from Y}{\tilR^G(\chall)=1}\leq\ppr{\chall\from Y}{R^G(\chall)=1}+\ppr{\chall\from Y}{\tilR^G(\chall)\ne R^G(\chall)}$$
		and that
		$$\ppr{\chall\from Y}{\tilR^G(\chall)=1}\geq\ppr{\chall\from Y}{\tilR^G(\chall)=R^G(\chall)}\cdot\ppr{\chall\from Y}{R^G(\chall)=1}.$$
		Let $\widetilde{\funcSpace}$ denote the set of all functions in $\funcSpace_{n,n+\stretch}$ such that
		$$\ppr{\chall\from Y}{\tilR^G(\chall)=1}\geq\ppr{\chall\from Y}{R^G(\chall)=1}.$$
		Then
		\begin{align*}
			&\eex{G\from\funcSpace_{n,n+\stretch}}{\abs{\ppr{\chall\from Y}{\tilR^G(\chall)=1}-\ppr{\chall\from Y}{R^G(\chall)=1}}}\\
			&\quad\leq\frac{1}{2^{2^n(n+\stretch)}}\cdot\sum_{G\in\widetilde{\funcSpace}}\paren{\ppr{\chall\from Y}{\tilR^G(\chall)=1}-\ppr{\chall\from Y}{R^G(\chall)=1}}\\
			&\qquad+\frac{1}{2^{2^n(n+\stretch)}}\cdot\sum_{G\in\funcSpace_{n,n+\stretch}\setminus\widetilde{\funcSpace}}\paren{\ppr{\chall\from Y}{R^G(\chall)=1}-\ppr{\chall\from Y}{\tilR^G(\chall)=1}}\\
			&\quad\leq\frac{1}{2^{2^n(n+\stretch)}}\cdot\sum_{G\in\widetilde{\funcSpace}}\ppr{\chall\from Y}{\tilR^G(\chall)\ne R^G(\chall)}\\
			&\qquad+\frac{1}{2^{2^n(n+\stretch)}}\cdot\sum_{G\in\funcSpace_{n,n+\stretch}\setminus\widetilde{\funcSpace}}\ppr{\chall\from Y}{R^G(\chall)=1}\cdot\ppr{\chall\from Y}{\tilR^G(\chall)\ne R^G(\chall)}\\
			&\quad\leq\eex{G\from\funcSpace_{n,n+\stretch}}{\ppr{\chall\from Y}{\tilR^G(\chall)\ne R^G(\chall)}}\\
			&\quad\leq\numcallsdist\cdot e^{\keylen-\frac{\delta^2}{2}\cdot\nu}+5\numcallsdist\delta,
		\end{align*}
		where the last inequality follows from \cref{eq:out_same}.
	\end{proof}

	Now, observe that
	\ifdefined\IsLLNCS
	\begin{align*}
		\Delta\leq
		&\Pr_{G\from\funcSpace_{n,n+\stretch}}
		\left[
		\begin{aligned}
			&\left|\ppr{s\from\zon}{R^{G}(1^n,G(s))=1}\right.\\
			&\qquad\left.-\ppr{y\from\zo^{n+\stretch}}{R^{G}(1^n,y)=1}\right|<\frac{1}{n^{\gamma}}
		\end{aligned}
		\right]\\
		&+\Pr_{G\from\funcSpace_{n,n+\stretch}}
		\left[
		\begin{aligned}
			&\left|\ppr{s\from\zon}{R^{G}(1^n,G(s))=1}\right.\\
			&\qquad\left.-\ppr{s\from\zon}{\tilR^{G}(1^n,G(s))=1}\right|>\frac{1}{3n^{\gamma}}
		\end{aligned}
		\right]\\
		&+\Pr_{G\from\funcSpace_{n,n+\stretch}}
		\left[
		\begin{aligned}
			&\left|\ppr{y\from\zo^{n+\stretch}}{R^{G}(1^n,y)=1}\right.\\
			&\qquad\left.-\ppr{y\from\zo^{n+\stretch}}{\tilR^{G}(1^n,y)=1}\right|>\frac{1}{3n^{\gamma}}
		\end{aligned}
		\right].
	\end{align*}
	\else
	\begin{align*}
		\Delta\leq
		&\ppr{G\from\funcSpace_{n,n+\stretch}}{\abs{\ppr{s\from\zon}{R^{G}(1^n,G(s))=1}-\ppr{y\from\zo^{n+\stretch}}{R^{G}(1^n,y)=1}}<\frac{1}{n^{\gamma}}}\\
		&+\ppr{G\from\funcSpace_{n,n+\stretch}}{\abs{
				\ppr{s\from\zon}{R^{G}(1^n,G(s))=1}-
				\ppr{s\from\zon}{\tilR^{G}(1^n,G(s))=1}
			}> \frac{1}{3n^{\gamma}}}\\
		&+\ppr{G\from\funcSpace_{n,n+\stretch}}{\abs{
				\ppr{y\from\zo^{n+\stretch}}{R^{G}(1^n,y)=1}-
				\ppr{y\from\zo^{n+\stretch}}{\tilR^{G}(1^n,y)=1}} > \frac{1}{3n^{\gamma}}}.
	\end{align*}
	\fi
	Note that by \cref{eq:R_dist}, the first term is 0 for infinitely many $n$'s. As for the other two terms, by Markov's inequality and \cref{clm:exp_diff}, it follows that both terms are upper-bounded by
	$$3n^\gamma\cdot\paren{\numcallsdist\cdot e^{\keylen-\frac{\delta^2}{2}\cdot\nu}+5\numcallsdist\delta}.$$
	We conclude that 
	$$\Delta\leq6n^\gamma\cdot\paren{\numcallsdist\cdot e^{\keylen-\frac{\delta^2}{2}\cdot\nu}+5\numcallsdist\delta}$$
	for infinitely many $n$'s.
	
	We next choose the parameters such that $\Delta\leq3/4$ for infinitely many $n$'s. Specifically, we choose $\delta$ and $\nu$ (which we later translate to choosing $\numinputquery_t$ for every $t\in[c]$). Let $\delta=1/{48\numcallsdist n^\gamma\sqrt[4]\nu}$ and let $\nu=2(\keylen+n)\cdot\numcallsdist^2 n^{2\gamma}$. Then $\delta\geq2/\nu$ and $\nu\geq2(\keylen+n)/\delta^2$ if $\nu$ is sufficiently large. In this case, $e^{\keylen-\frac{\delta^2}{2}\cdot\nu}\leq e^{-n}$. Since $\nu$ is a polynomial of $n$, $\nu\leq 2\numcallsdist e^{n}$ for all sufficiently large $n$'s. Therefore $\delta\geq2/\nu\geq e^{-n}/\numcallsdist$, hence $e^{-n}\leq \numcallsdist\delta$. Thus, for such choices of $\nu$, $\Delta\leq 36\numcallsdist n^\gamma\cdot\delta\leq3/4$.
	
	Now, recall that $\nu=\min_{k\in\keyDom}|\lastset_k|$, and that for every $k\in\keyDom$, either $|\lastset_k|=\numinputquery_c$ or $|\lastset_k|=\frac{\numinputquery_{t-1}}{(c-t)\numinputquery_t}$ for some $t\in[c]$. Further recall that for $\adv_\ideal$ to be have $1-e^{-n}$ advantage, we took $\numinputquery_c\geq 10(\keylen+c(n+\stretch)+n)$ and $\numinputquery_{t-1}\geq (c-t+1)\numinputquery_t\cdot 10(\keylen+c(n+\stretch)+n)$ for all $t\in[c]$. Therefore, we let\footnote{The bound on $\numinputquery_c$ can be improved to $10(\keylen+c(n+\stretch)+n)$. This is because the other term comes from considering $k$'s such that $\recSteps(k)<c$ (recall that if $\recSteps(k)=c$ then both scores are identical).}
	$$\numinputquery_c=10(\keylen+c(n+\stretch)+n)+2(\keylen+n)\cdot\numcallsdist^2 n^{2\gamma}$$
	and
	$$\numinputquery_{t-1}=(c-t+1)\paren{10(\keylen+c(n+\stretch)+n)+2(\keylen+n)\cdot\numcallsdist^2 n^{2\gamma}}\cdot \numinputquery_t$$
	for every $t\in[c]$. {As a result, the adversaries query the function on $\numinputquery=\numinputquery_0\leq n^{\alpha c}\cdot c!$ inputs, for some constant $\alpha\in\NN$.\footnote{$\alpha$ depends on the degree of the polynomials $\keylen$, $\stretch$, $\numcallsdist$, and $\gamma$.}. Since there are $2^{\inlen}$ possible inputs, for the adversaries to be well-defined, it must be the case where $n^{\alpha c}\cdot c!\leq 2^{\inlen}$. Indeed, since $c=\frac{\inlen}{w(n)\cdot \log\inlen}$ for some $w(n)=\omega(1)$,
	$$n^{\alpha c}\cdot c!\leq n^{\alpha c}\cdot c^c=2^{\alpha c\log n+c\log c}\leq2^{\alpha \frac{\inlen}{w(n)\log \inlen}\log n+\frac{\inlen}{w(n)\log \inlen}\log\inlen}\leq 2^{\inlen},$$
	for all sufficiently large $n$'s, where the last inequality is due to the fact that $\inlen=\poly(n)$, hence $\log n=O(\log\inlen)$.}
	
	Finally, in order to contradict \cref{lem:RO_is_PRG}, we show that the total number of queries $\tilR$ makes to $G$ is at most $2^{n/4}$. Indeed, in each of the $\numcallsdist$ calls to the distinguishing oracle, the number of calls to $G$ that the reduction makes is at most 
	$$\poly(m)\leq n^{\beta c}\cdot (c!)^\beta\leq n^{\beta c}\cdot c^{\beta c}=2^{\beta c\log n+\beta c\log c}$$
	for some constant $\beta\in\NN$.
	Recall that $c=\frac{n}{w(n)\cdot \log n}$, where $w(n)=\omega(1)$. Therefore, the total number of calls that $\tilR$ makes is at most
	\begin{align*}
		d\cdot2^{\beta c\log n+\beta c\log c}
		=d\cdot2^{\frac{\beta n}{w(n)}+\frac{\beta n}{w(n)\log n}\log\frac{n}{w(n)\log n}}
		\leq d\cdot2^{\frac{2\beta n}{w(n)}}
        \leq 2^{n/4},
	\end{align*}
	for sufficiently large $n$'s.
\end{proof}

\begin{remark}[On the definition of black-box reduction]
	In the definition we considered, the advantage of the reduction $R$ is a polynomial of the advantage of the distinguishing oracle $\adv$. This in particular allowed us to conclude the existence of $\gamma\in\NN$ for which \cref{eq:R_dist} holds. In the definition for black box constructions of PRF from \citet{BMM24}, it is only required that if $\adv$ succeeds with noticeable probability, then so does $R$. Under this definition, \cref{eq:R_dist} does not necessarily hold, since the advantage of $R$ might decrease as a function of the advantage of $\adv$.
	
	Nevertheless, to the best of our knowledge, all black-box reductions studied in the literature satisfy the stronger requirement. Moreover, our proof continues to hold under the weaker assumption that the advantage of $R$ is bounded below for all distinguishing oracles with a sufficiently large advantage.
\end{remark}
\section{Lower Bounds for Black-Box Constructions of a PRF With a Long Output}\label{sec:LB_PRF}
In this section, we show our lower bounds for fully black-box constructions of PRF with a long output from a PRG, with adaptive calls.
We will use the following result, stating the existence of a PRG \wrt some oracle $f$.
\begin{lemma}[{\cite[Lemma 3.5]{BMM24}}]\label{lem:PRG_wrt_oracle}
	For every oracle $f=\set{f_n:\zon\to\zon}_{n\in\NN}$ and $w,r:\NN\to\NN$ such that $n+r(n)=2^{o(w(n))}$, there exists a (possibly inefficient) function $G=\set{G_n:\zo^{w(n)}\to\zo^{n+r(n)}}_{n\in\NN}$ such that the following holds. For every \ppt oracle-aided algorithm $A$, there exists a negligible function $\mu$ such that
	$$\abs{\ppr{s\from\zo^{w(n)}}{A^{f,G}\of{1^n,G\of{1^n,s}}=1}-\ppr{y\from\zo^{n+r(n)}}{A^{f,G}\of{1^n,y}=1}}\leq\mu(n).$$
\end{lemma}

In what follows, we refer to the function $G$ as the PRG \wrt the oracle $f$.
We are now ready to state and prove our lower bound on the number of calls to the PRG required to construct a PRF with a long output. Throughout the section we fix the input length $\inlen=\inlen(n)$, the output length $\outlen=\outlen(n)$, the key length $\keylen=\keylen(n)$, and the PRG stretch $\stretch=\stretch(n)$.

\begin{theorem}\label{thm:PRF_from_PRG_gen}
	For every $c(n)=\frac{\outlen}{\stretch+\omega(\log n)}$, there is no $c(n)$-call fully black-box construction of a weak PRF with output length $\outlen$ from an $\stretch$-bit stretch PRG.
\end{theorem}
\begin{proof}
	Fix $w=w(n)=\omega(\log n)$ and let $\numquery=\numquery(n)=\outlen/(\stretch+w)$. Assume towards contradiction that there exists a $\numquery$-call fully black-box construction $(\prf,R)$ from an $\stretch$-bit stretch PRG, where
	$$\prf=\set{f^{(\cdot)}_k:\zo^\inlen\to\zo^\outlen}_{n\in\NN,k\in\zo^\keylen}.$$
	That is, for every $i\in[\numquery]$ there exist $S_i:\zo^{\inlen}\times\zo^{\keylen}\times\sparen{\zo^{(n+\stretch)}}^{i-1}\to\zon$, and there exists $P:\zom\times\zo^{\keylen}\times\paren{\zo^{n+\stretch}}^\numquery\to\zo^\outlen$ such that
	$$f^G_k(x)=P\of{x,k,G\of{s_1},\ldots,G(s_\numquery)}$$
	for all $G:\zon\to\zo^{n+\stretch}$, where $s_i=S_i(x,k,G\of{s_1},\ldots,G(s_{i-1}))$ for all $i\in[\numquery]$.
	
	To obtain a contradiction, we consider a PRG that is applied only to the first $w$ bits of the seed. This, intuitively, implies that the entropy of the output of the PRF over $\numinputquery$ inputs is provided only by the key (which has $\keylen$ entropy) and $\numinputquery\cdot (w+\stretch)\cdot\numquery$ additional bits from all applications of the $\stretch$-bit stretch PRG. On the other hand, for a random function, the amount of entropy is $\numinputquery\cdot\outlen$. Taking $\numinputquery=(n+\keylen)/\numquery$ results in a large gap between the two cases.

    In the following, for a string $s=(s_1,\ldots,s_n)\in\zon$ and $i\in[n]$ we let $s_{< i}=(s_1,\ldots,s_{i-1})$ and let $s_{\geq i}=(s_{i},\ldots,s_n)$. We use $\conc$ to denote the concatenation of strings. 
	
	We first define a distinguisher $\Break^{(\cdot)}$ for the PRF.
	
	\begin{algorithm}[$\Break$]~\\
		{Input:} The security parameter $1^n$.\\
		{Oracle access:} A function $f:\zo^{\inlen}\to\zo^{\outlen}$.
		
		\begin{enumerate}[topsep=0pt]
			\item Let $\numinputquery=\numinputquery(n)=(n+\keylen)/\numquery$ and let $x_1,\ldots,x_\numinputquery\in\zo^\inlen$ denote the lexicographically first $\numinputquery$ inputs.\footnote{Note that if we obtain a contradiction for all distinct inputs, then the same holds for random inputs, up to the birthday bound.}
			
			\item Query $f$ on $x_1,\ldots,x_\numinputquery$, and let $z_1,\ldots,z_\numinputquery\in\zo^{\outlen}$ denote the respective answers.
			
			\item Output 1 if and only if there exists $k\in\zo^{\keylen}$ and $(y_{i,j})_{i\in[\numinputquery],j\in[\numquery]}\in\sparen{\zo^{w+\stretch-1}}^{\numinputquery\cdot\numquery}$ such that
			$$P\paren{x_i,k,\paren{y_{i,1}\conc s_{i,1}},\ldots,\paren{y_{i,\numquery}\conc s_{i,\numquery}}}=z_i$$
			for all $i\in[\numinputquery]$, where
			$$s_{i,j}=S_j\of{x_i,k,\paren{y_{i,1}\conc s_{i,1}},\ldots,\paren{y_{i,j-1}\conc s_{i,j-1}}}_{\geq\frac{\outlen}{\numquery}-\stretch}$$
			For all $i\in[\numinputquery]$ and $j\in[\numquery]$.
		\end{enumerate}
	\end{algorithm}
	
	We next define a PRG for which $\Break$ breaks the security of $\prf$. Since $w=\omega(\log n)$ and $\stretch$ is a polynomial, it follows that $n+\stretch=2^{o(w-1)}$. Thus, by \cref{lem:PRG_wrt_oracle} there exists a PRG $G':\zo^{w-1}\to\zo^{w+\stretch-1}$ \wrt the oracle $\Break$, that cannot be broken with non-negligible probability. Consider the PRG $G:\zon\to\zo^{n+\stretch}$ defined as
	$$G(s)=G'\of{s_{<w}}\conc s_{\geq w}.$$
	Observe that $\Break^{f^G_k}$ always outputs 1. To reach a contradiction, we next show that for a random function $f$, this probability is at most $2^{-n}$. Indeed, by the union bound,
	\ifdefined\IsLLNCS
	\begin{align*}
		&\ppr{f\from\funcSpace_{\inlen}}{\Break^{f}(1^n)=1}\\
		&\quad=\Pr_{f\from\funcSpace_{\inlen}}\left[
		\begin{aligned}
			&\exists k\in\zo^{\keylen}\;\exists \paren{y_{i,j}}_{i\in[\numinputquery],j\in[\numquery]}\in\paren{\zo^{w+\stretch-1}}^{\numinputquery\cdot\numquery}\;\forall i\in[\numinputquery]:\\
			&\quad P\paren{x_i,k,\paren{y_{i,j}\conc s_{i,j}}_{j=1}^{\numquery}}=f(x_i)
		\end{aligned}
		\right]\\
		&\quad\leq2^{\keylen+\numinputquery\cdot\numquery\cdot\paren{w+\stretch-1}}\cdot\ppr{f\from\funcSpace_{\inlen}}{\forall i\in[\numinputquery]: P\paren{x_i,k,\paren{y_{i,j}\conc s_{i,j}}_{j=1}^{\numquery}}=f(x_i)}\\
		&\quad=\frac{2^{\keylen+\numinputquery\cdot\outlen-\numinputquery\cdot\numquery}}{2^{\numinputquery\cdot\outlen}}\\
		&\quad=2^{-n}.
	\end{align*}
	\else
	\begin{align*}
		&\ppr{f\from\funcSpace_{\inlen,\outlen}}{\Break^{f}(1^n)=1}\\
		&\quad=\ppr{f\from\funcSpace_{\inlen,\outlen}}{\exists k\in\zo^{\keylen}\;\exists \paren{y_{i,j}}_{i\in[\numinputquery],j\in[\numquery]}\in\paren{\zo^{w+\stretch}}^{\numinputquery\cdot\numquery}\;\forall i\in[\numinputquery]: P\paren{x_i,k,\paren{y_{i,j}\conc s_{i,j}}_{j=1}^{\numquery}}=f(x_i)}\\
		&\quad\leq2^{\keylen+\numinputquery\cdot\numquery\cdot\paren{w+\stretch-1}}\cdot\ppr{f\from\funcSpace_{\inlen}}{\forall i\in[\numinputquery]: P\paren{x_i,k,\paren{y_{i,j}\conc s_{i,j}}_{j=1}^{\numquery}}=f(x_i)}\\
		&\quad=\frac{2^{\keylen+\numinputquery\cdot\outlen-\numinputquery\cdot\numquery}}{2^{\numinputquery\cdot\outlen}}\\
		&\quad=2^{-n}.
	\end{align*}
	\fi
    Since $(\prf,R)$ is assumed to be a fully black-box PRF construction, it follows that $R^{\adv,G}$ breaks $G$. However, this implies the existence of a \ppt algorithm $\tilR^{\adv,G'}$ breaking $G'$, contradicting \cref{lem:PRG_wrt_oracle}.
\end{proof}

\bibliographystyle{abbrvnat}
\bibliography{crypto}

\end{document}